\documentclass{article}%

\newcommand{\sv}[1]{}%
 \newcommand{\lv}[1]{#1}%
\usepackage{etoolbox}
\newcommand{\appendixText}{}

 \newcommand{\toappendix}[1]{#1}%

\lv{
\usepackage[utf8]{inputenc}
\usepackage{graphicx} %
\usepackage{subcaption}
\usepackage{lmodern}
\usepackage{a4wide}
\usepackage{authblk}
}

\usepackage{amssymb,amsmath,graphicx}
\usepackage{color,enumerate,comment,boxedminipage}
\usepackage{pgf,float,hyperref}
\usepackage{xcolor}
\definecolor{darkblue}{rgb}{0,0,0.45}
\definecolor{darkred}{rgb}{0.6,0,0}
\definecolor{darkgreen}{rgb}{0.13,0.5,0}
\hypersetup{colorlinks, linkcolor=darkblue, citecolor=darkgreen,
urlcolor=darkblue}
\usepackage[T1]{fontenc}
\usepackage{xspace}
\usepackage{paralist}
\usepackage{cleveref}
\usepackage{enumitem}
\usepackage[sort&compress,numbers]{natbib}

\newcommand{\Oh}{\mathcal{O}}
\newcommand{\NP}{\textsf{NP}}
\newcommand{\N}{\mathbb{N}}

\newcommand{\cV}{\mathcal{V}}
\newcommand{\cE}{\mathcal{E}}
\newcommand{\cG}{\mathcal{G}}
\newcommand{\cF}{\mathcal{F}}
\newcommand{\cS}{\mathcal{S}}

\newcommand{\col}[1]{#1-\textsc{Coloring}\xspace}
\newcommand{\lcol}[1]{\textsc{List} #1-\textsc{Coloring}\xspace}
\newcommand{\lhom}[1]{\textsc{LHom}(#1)\xspace}
\newcommand{\lshom}[1]{\textsc{LSHom}(#1)\xspace}
\newcommand{\llshom}[1]{\textsc{LLSHom}(#1)\xspace}
\newcommand{\nae}{\textsc{NAE 3-Sat}\xspace}
\newcommand{\cHpoly}{\mathcal{H}_{\mathrm{poly}}}

\newcommand{\vrb}{\textit{Vrb}}
\newcommand{\cls}{\textit{Cls}}
\newcommand{\yes}{\texttt{yes}}

\renewcommand{\tilde}{\widetilde}
\newcommand{\tos}{\xrightarrow{s}}
\renewcommand{\cref}{\Cref}

\lv{\usepackage[absolute]{textpos}}

\sv{

}
\usepackage{amsthm}

\sv{
\newtheorem{claimm}{Claim}[section]
\numberwithin{claimm}{section}

\crefname{claimm}{Claim}{Claims}
}

\lv{
\newtheorem{lemma}{Lemma}[section]
\newtheorem{theorem}[lemma]{Theorem}

\newtheorem{corollary}[lemma]{Corollary}
\newtheorem{claimm}[lemma]{Claim}
\newtheorem{definition}[lemma]{Definition}
}

\sv{
\newtheorem*{lemma*}{Lemma}
\newtheorem*{theorem*}{Theorem}
\newtheorem*{corollary*}{Corollary}
}

\sv{
\usepackage [pagewise]{lineno}
\linenumbers
}

\usepackage{todonotes}

\setdescription{itemsep=0pt,parsep=0pt,topsep=0pt,leftmargin=0.5cm}

\sv{
\title{List Locally Surjective Homomorphisms// in Hereditary Graph Classes}

\author{Pavel Dvo\v{r}\'{a}k\inst{1}\lv{\thanks{Koblich's pile of gold}}
\and
Monika Krawczyk\inst{2}
\and
Tom\'{a}\v{s} Masa\v{r}\'ik\inst{3}\lv{\thanks{Received funding from the European Research Council (ERC) under the European Union’s Horizon 2020 research and innovation programme Grant Agreement 948057.}}
\and\\
Jana Novotn\'a\inst{3,4}\lv{\thanks{Supported by SVV-2020–260578 and GAUK 384321 of Charles University and by funding from the European Research Council (ERC) under the European Union’s Horizon 2020 research and innovation programme Grant Agreement 71470.}}
\and
Pawe{\l} Rz{\k{a}}\.{z}ewski\inst{2,3}\lv{\thanks{Supported by Polish National Science Centre grant no. 2018/31/D/ST6/00062.}}
\and
Aneta \.Zuk\inst{2}
}

\institute{
Department of Computer Science, University of Bristol, UK\newline
\email{koblich@iuuk.mff.cuni.cz}
\and
Faculty of Mathematics and Information Science,\\Warsaw University of Technology, 
Warsaw, Poland\\
\email{p.rzazewski@mini.pw.edu.pl}
\and 
Institute of Informatics, Faculty of Mathematics, Informatics and Mechanics, University of Warsaw, Poland
\email{masarik@mimuw.edu.pl}
\and
Department of Applied Mathematics, Faculty of Mathematics and Physics, Charles University, Prague, Czech Republic
\email{janca@kam.mff.cuni.cz}
}
}

\lv{%
  \title{List Locally Surjective Homomorphisms in Hereditary Graph Classes\thanks{%
  P.D.~is supported by EPSRC New Investigator Award EP/V010611/1.
T.M.~received funding from the European Research Council (ERC) under the European Union’s Horizon 2020 research and innovation programme Grant Agreement 948057.
J.N.~was supported by SVV-2020–260578 and GAUK 384321 of Charles University.
P.Rz.~was supported by Polish National Science Centre grant no.~2018/31/D/ST6/00062.
    }
}

\author[1]{Pavel Dvořák}
\author[2]{Monika Krawczyk}
\author[3]{Tomáš Masařík}
\author[3,4]{\\Jana Novotná}
\author[2,3]{Paweł Rzążewski}
\author[2]{Aneta Żuk}

\affil[1]{Department of Computer Science, University of Bristol, UK}
\affil[ ]{\texttt{koblich@iuuk.mff.cuni.cz}}
\affil[2]{Faculty of Mathematics and Information Science, Warsaw University of Technology, Warsaw, Poland}
\affil[ ]{\texttt{p.rzazewski@mini.pw.edu.pl}}
\affil[3]{Institute of Informatics, Faculty of Mathematics, Informatics and Mechanics, University of Warsaw, Poland}
\affil[ ]{\texttt{masarik@mimuw.edu.pl}}
\affil[4]{Department of Applied Mathematics, Faculty of Mathematics and Physics, Charles University, Prague, Czech Republic}
\affil[ ]{\texttt{janca@kam.mff.cuni.cz}}
\date{}
}

\begin{document}

\maketitle
\lv{%
\begin{textblock}{20}(0, 14.0)
\includegraphics[width=40px]{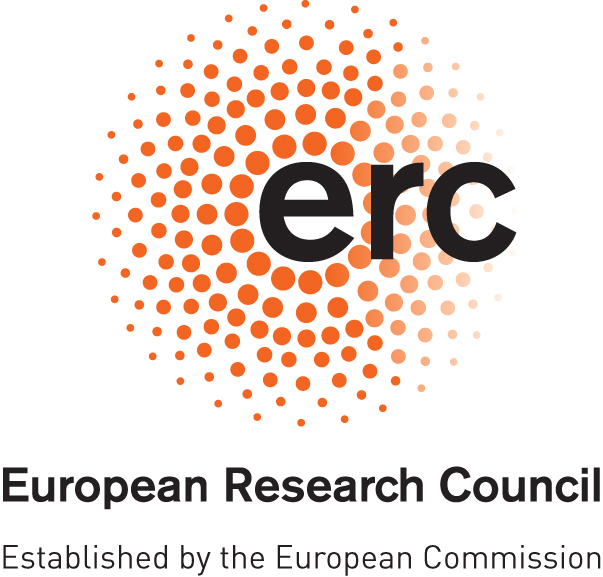}%
\end{textblock}
\begin{textblock}{20}(0, 14.9)
\includegraphics[width=40px]{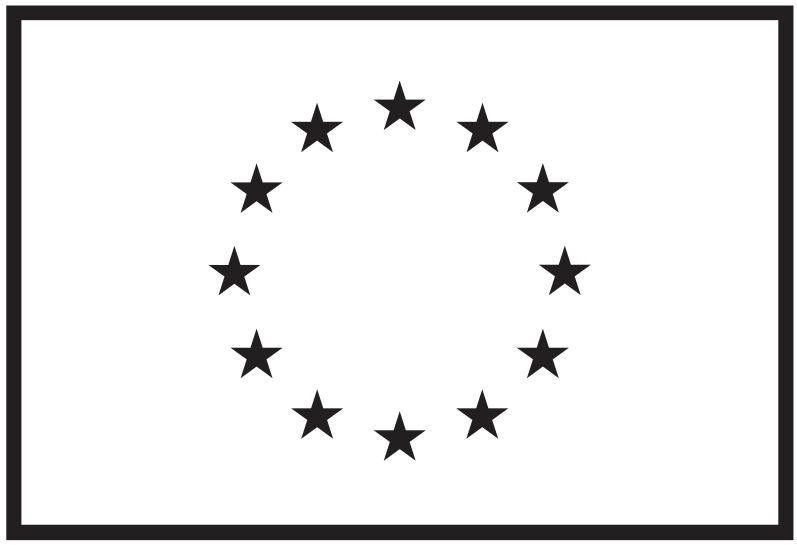}%
\end{textblock}
}

\begin{abstract}
A \emph{locally surjective homomorphism} from a graph $G$ to a graph $H$ is an edge-preserving mapping from $V(G)$ to $V(H)$ that is surjective in the neighborhood of each vertex in $G$.
In the \emph{list locally surjective homomorphism} problem, denoted by \textsc{LLSHom}($H$), the graph $H$ is fixed and the instance consists of a graph $G$ whose every vertex is equipped with a subset of $V(H)$, called list.
We ask for the existence of a locally surjective homomorphism from $G$ to $H$, where every vertex of $G$ is mapped to a vertex from its list.
In this paper, we study the complexity of the \textsc{LLSHom}($H$) problem in $F$-free graphs, i.e., graphs that exclude a fixed graph $F$ as an induced subgraph.
We aim to understand for which pairs $(H,F)$ the problem can be solved in subexponential time.

We show that for all graphs $H$, for which the problem is \textsf{NP}-hard in general graphs, it cannot be solved in subexponential time in $F$-free graphs for $F$ being a bounded-degree forest, unless the ETH fails.
The initial study reveals that a natural subfamily of bounded-degree forests $F$, that might lead to some tractability results, is the family $\mathcal{S}$ consisting of forests whose every component has at most three leaves.
In this case, we exhibit the following dichotomy theorem: besides the cases that are polynomial-time solvable in general graphs, the graphs $H \in \{P_3,C_4\}$ are the only connected ones that allow for a subexponential-time algorithm in $F$-free graphs for every $F \in \mathcal{S}$ (unless the ETH fails).

\end{abstract}

\section{Introduction}
Graph coloring is arguably one of the best-studied problems in algorithmic graph theory.
It is well-known that \col{$k$} is polynomial-time solvable for $k \leq 2$ and \NP-hard for every $k \geq 3$~\cite{DBLP:conf/coco/Karp72}.
Furthermore, assuming the Exponential-Time Hypothesis (ETH)~\cite{ImpagliazzoPaturi,DBLP:journals/jcss/ImpagliazzoPZ01}, the hard cases do not even admit algorithms working in subexponential time.

\paragraph*{Coloring $F$-free graphs.}
A very natural direction of research is to investigate what restrictions put on the family of input make the problem
more tractable than for general graphs. 
In recent years, a very active topic has been to study the complexity of \col{$k$} and related problems in graphs defined by one or more forbidden induced subgraphs.
For a family $\cF$ of graphs, we say that a graph $G$ is $\cF$-free if $G$ does not contain any graph from $\cF$ as an induced subgraph.
If $\cF$ consists of a single graph $F$, then we say $F$-free instead of $\{F\}$-free.
Note that the class of  $\cF$-free graphs is \emph{hereditary} i.e., closed under vertex deletion.
On the other hand, every hereditary class of graphs can be equivalently defined as $\cF$-free graphs for some unique minimal (possibly infinite) family $\cF$ of graphs.

It is well-known that for every $k \geq 3$, the \col{$k$} problem is \NP-hard in $F$-free graphs, unless $F$ is a \emph{linear forest}, i.e., every connected component of $F$ is a path.
Indeed, for every constant $g$, \col{$k$} is \NP-hard in graphs of girth (i.e., the length of a shortest cycle) at least $g$~\cite{DBLP:journals/cpc/Emden-WeinertHK98}. Setting $g = |V(F)|+1$, we immediately obtain hardness for every $F$ that is not a forest. On the other hand, \col{$k$} is \NP-hard in line graphs which are claw-free~\cite{DBLP:journals/siamcomp/Holyer81a,DBLP:journals/jal/LevenG83}. The only forests that are claw-free are linear forests.

The complexity of \col{$k$} in $P_t$-free graphs, where $P_t$ is the path in $t$ vertices, has recently attracted a lot of attention.
For $t =5$, the problem is polynomial-time solvable for every constant $k$~\cite{DBLP:journals/algorithmica/HoangKLSS10}.
If $k \geq 5$, then the problem is \NP-hard already in $P_6$-free graphs~\cite{DBLP:journals/ejc/Huang16}.
The case $k=4$ is also fully understood: it is polynomial-time solvable for $t \leq 6$~\cite{DBLP:conf/soda/SpirklCZ19} and \NP-hard for $t \geq 7$~\cite{DBLP:journals/ejc/Huang16}.
The case of $k=3$ is much more elusive. We know a polynomial-time algorithm for $P_7$-free graphs~\cite{DBLP:journals/combinatorica/BonomoCMSSZ18}. However, for $t \geq 8$,
we know neither polynomial-time algorithm nor any hardness result.
Some positive results are also known for the case that $F$ is a disconnected linear forest~\cite{DBLP:journals/algorithmica/KlimosovaMMNPS20,DBLP:journals/algorithmica/ChudnovskyHSZ21,DBLP:journals/corr/Golovach0PS14}.

Let us point out almost all mentioned algorithmic results also hold for the more general \emph{list} variant of the problem, 
where each vertex is given a list of admissible colors. The notable exception is \lcol{4}, which is \NP-hard already in $P_6$-free graphs~\cite{DBLP:journals/iandc/GolovachPS14}. Furthermore, all hardness results also imply the nonexistence of subexponential-time algorithms (assuming the ETH).

Some more general positive results can be obtained if we relax our notion of tractability.
As observed by Groenland et al.~\cite{groenland2019h}, \lcol{3} can be solved in subexponential time in $P_t$-free graphs, for every fixed $t$.
This was recently improved by Pilipczuk, Pilipczuk, and Rzążewski~\cite{DBLP:conf/sosa/PilipczukPR21} who showed a quasipolynomial-time algorithm for this problem.
Note that this is strong evidence that the problem is not \NP-hard.

\paragraph*{Graph homomorphisms.} Graph colorings can be seen as a special case of \emph{graph homomorphisms}.
A homomorphism from a graph $G$ to a graph $H$ (with possible loops) is an edge-preserving mapping from $V(G)$ to $V(H)$.
Note that homomorphisms to $K_k$ are precisely proper $k$-colorings.
By the celebrated result of Hell and Ne\v{s}et\v{r}il~\cite{DBLP:journals/jct/HellN90}, determining whether an input graph $G$ admits a homomorphism to a fixed graph $H$ is polynomial-time solvable if $H$ is bipartite or has a vertex with a loop, and \NP-hard otherwise.
A list variant of the graph homomorphism problem, denoted by \lhom{$H$}, has also been considered.
It turns out that the problem can be solved in polynomial time if $H$ is a so-called \emph{bi-arc graph}, and otherwise, the problem is \NP-hard~\cite{FEDER1998236,DBLP:journals/combinatorica/FederHH99,DBLP:journals/jgt/FederHH03}.

The complexity of variants of the graph homomorphism problem in hereditary graph classes was also studied.
For example, Chudnovsky et al.~\cite{DBLP:conf/esa/ChudnovskyHRSZ19} showed that \lhom{$C_k$} for $k \in \{5,7\} \cup [9,\infty)$ is polynomial-time solvable in $P_9$-free graphs. On the negative side, they showed that for every $k \geq 5$ the problem cannot be solved in $F$-free graphs,
unless every component of $F \in \cS$, where $\cS$ consists of graphs whose every connected component is a path or a tree with three leaves (called a \emph{subdivided claw}).
This negative result was later extended by Piecyk and Rzążewski~\cite{DBLP:conf/stacs/PiecykR21} who showed that if $H$ is not a bi-arc graph (i.e., \lhom{$H$} is \NP-hard in general graphs), then \lhom{$H$} is \NP-hard and cannot be solved in subexponential time (assuming the ETH) in $F$-free graphs,
unless $F \in \cS$.

The case of forbidden path or subdivided claw was later investigated by Okrasa and Rzążewski~\cite{DBLP:conf/stacs/OkrasaR21}.
They defined a class of \emph{predacious graphs}
and showed that if $H$ is not predacious, then for every $H$, the \lhom{$H$} problem
can be solved in quasipolynomial time in $P_t$-free graphs (for every $t$). Otherwise, for every $H$, there exists $t$
for which \lhom{$H$} cannot be solved in subexponential time in $P_t$-free graphs unless the ETH fails.
They also provided some partial results for the case of forbidden subdivided claws.

The complexity of variants of the graph homomorphism problem in other hereditary graph classes has also been considered~\cite{DBLP:journals/dm/FederHH07,DBLP:journals/siamdm/ChudnovskyKPRS21,DBLP:journals/jcss/OkrasaR20}.

\paragraph*{Locally surjective graph homomorphisms.} Graph homomorphisms are a very robust notion, which can be easily extended by putting some additional restrictions on the solution. In this paper, we focus on one such variant called \emph{locally surjective homomorphisms}.
A homomorphism $h$ from $G$ to $H$ is locally surjective if it is surjective in the neighborhood of each vertex of $G$.
In other words, if $h(v)=a \in V(H)$, then for every neighbor $b$ of $a$ in $H$ (including $a$, if it has a loop) there is a neighbor $v'$ of $v$ in $G$, such that $h(v')=b$.
The study of locally surjective homomorphisms originates in social sciences, where they can be used to model some social roles (the problem is called \emph{role assignment}~\cite{EVERETT1991183}).
The problem of determining whether an input graph admits a locally surjective homomorphism to a fixed graph $H$ is denoted by \lshom{$H$}.
Fiala and Paulusma~\cite{DBLP:journals/tcs/FialaP05} provided the full complexity dichotomy for \lshom{$H$}.
For simplicity, let us consider only connected graphs $H$, and let $K_1^\circ$ be the one-vertex graph with a loop.
They showed that \lshom{$H$} is polynomial-time-solvable if $H \in \cHpoly := \{K_1,K_1^\circ,K_2\}$, and otherwise it is \NP-hard.
Again, the hardness reduction excluded also subexponential time algorithms under the ETH.

Let us point out that \lshom{$P_3$} is closely related to the well-known hypergraph 2-coloring problem~\cite{Lovasz73} (or, equivalently, \textsc{Positive NAE SAT}). In this problem, we ask whether the input hypergraph admits a 2-coloring of its vertices which makes no edge monochromatic. 
Consider a hypergraph $\bf H$ with vertices $\cV$ and hyperedges $\cE$, and let $G$ be its \emph{incidence graph},
i.e., the bipartite graph with vertex set $\cV \cup \cE$, where $v \in \cV$ is adjacent to $e \in \cE$ if and only if $v \in e$.
Note that  proper 2-colorings of  $\bf H$ are precisely locally surjective homomorphisms of $G$ to $P_3$ with consecutive vertices $1,2,3$,
where $\cV$ is mapped to $\{1,3\}$, and $\cE$ is mapped to $\{2\}$.
As shown by Camby and Schaudt~\cite{DBLP:journals/algorithmica/CambyS16}, 2-coloring of hypergraphs with $P_7$-free incidence graph is polynomial-time solvable.

The structural and computational aspects of locally surjective homomorphisms were studied by several authors~\cite{DBLP:journals/tcs/ChaplickFHPT15,DBLP:journals/csr/FialaK08,DBLP:journals/corr/abs-2201-11731}.
However, up to the best of our knowledge, no systematic study of \lshom{$H$} in hereditary graph classes has been conducted.

\paragraph*{Our contribution.}
In this paper, we consider the complexity of the \emph{list} variant of \lshom{$H$}, called \llshom{$H$}.
First, we observe that if $H \in \cHpoly$ (recall, these are the easy cases of \lshom{$H$}),
then also \llshom{$H$} can be solved in polynomial time in general graphs.

Then we focus on the complexity of the problem in $F$-free graphs.
In particular, we are interested in determining the pair $(H,F)$, for which the problem can be solved in subexponential time.
Similarly to the case of \lhom{$H$}, we split into two cases, depending whether $F \in \cS$.

In the first case, we identify two more positive cases: we show that if ${H \in \{P_3,C_4\}}$,
then the problem admits a subexponential-time algorithm for every $F \in \cS$.
The algorithm itself uses a win-win strategy: we combine branching on a high-degree vertex
with a separator theorem that can be used if the maximum degree is bounded.
A similar approach was used for various other problems~\cite{groenland2019h,DBLP:journals/algorithmica/NovotnaOPRLW21},
however, the specifics of our problem require a slightly more complicated approach.

We also show that the above cases are the only positive ones for general $F \in \cS$,
which provides the following dichotomy theorem.

\begin{theorem}\label{thm:pathfree}
Let $H \notin \cHpoly$ be a fixed connected graph.
\begin{enumerate}
\item If $H \in \{P_3,C_4\}$, then for every $F \in \cS$, the \llshom{$H$} problem can be solved in time $2^{\Oh((n \log n)^{2/3})}$ in $n$-vertex $F$-free graphs.
\item Otherwise there is $t$, such that  the \llshom{$H$} problem cannot be solved in subexponential time in $P_t$-free graphs, unless the ETH fails.
\end{enumerate}
\end{theorem}

Then, we turn our attention to other forbidden graphs $F$.
We show that whenever the problem is \NP-hard for general graphs,
i.e., for every $H \notin \cHpoly$, and for every $g >0$ there exists $d = d(H)$,
such that \lshom{$H$} is \NP-hard in graphs of degree at most $d$ and girth at least $g$.
This implies the following lower bound.\sv{\footnote{Proofs of statements marked with ($\spadesuit$) are postponed to the appendix. The proof of \cref{thm:OtherHardness} is presented in~\cref{sec:otherhard}.}}

\lv{\begin{theorem}\label{thm:OtherHardness}}
  \sv{\begin{theorem}[$\spadesuit$]\label{thm:OtherHardness}}
For every connected $H \notin \cHpoly$, there exists $d \in \N$, such that the following holds.
For every graph $F$ that is not a forest of maximum degree at most $d$,
the \llshom{$H$} problem cannot be solved in time $2^{o(n)}$ in $n$-vertex $F$-free graphs of maximum degree $d$, unless the ETH fails.
\end{theorem}

We conclude the paper by discussing the possibilities of improving our theorems in order to fully classify the complexity of \lshom{$H$} in $F$-free graphs.

\section{Preliminaries}\label{sec:prel}
For a graph $G$ and $v \in V(G)$, by $N_G(v)$ we denote the set of neighbors of $v$ in $G$.
For a set $X \subseteq V(G)$, by $N_G[X]$ we denote $X \cup \bigcup_{v \in X} N_G(v)$.
If $G$ is clear from the context, we omit the subscripts.
By $G[X]$, where $X \subseteq V(G)$, we denote the subgraph of $G$ induced by the set $X$.

For graphs $G$ and $H$, by $G \times H$ we denote their \emph{direct product} (sometimes called categorical product or Kronecker product), i.e, the graph
 \begin{align*}
  V(G \times H) &= V(G) \times V(H), \\
  E(G \times H) &= \bigl\{\{(u_1,v_1),(u_2,v_2)\} \mid u_1u_2 \in E(G) \land v_1v_2 \in E(H) \bigr\}.
 \end{align*}

For $t,a,b,c \geq 1$, by $P_t$ we denote the $t$-vertex path, and by $S_{a,b,c}$ we denote 
the three-leaf tree with leaves at distance $a$, $b$, and $c$, respectively,
from the unique vertex of degree 3, which we denote as \emph{central}. Every such $S_{a,b,c}$ is called a \emph{subdivided claw}.
Recall that by $\cS$, we denote the family of graphs whose every connected component
is either a path or a subdivided claw.

Let $h$ be a homomorphism from $G$ to $H$.
We say that a vertex $v \in V(G)$ is \emph{happy} (in $h$) if 
$h(N_G(v)) = N_H(h(v))$. 
In other words, for every neighbor $y$ of $h(v)$, some neighbor of $v$ is colored $y$.
We say that a homomorphism $h$ is \emph{locally surjective} if every vertex is happy in $h$.
If $h$ is a locally surjective homomorphism from $G$ to $H$, we denote it by $h: G \tos H$.

For a fixed graph $H$ (with possible loops) we consider the $\llshom{H}$ problem,
whose instance is $(G,L)$ where $G$ is a graph and $L : V(G) \to 2^{V(H)}$ is a list function.
We ask if there exists a homomrphism $h : G \tos H$, such that for every $v \in V(G)$ it holds that $h(v) \in L(v)$.
If $h$ is such a homomorphism, we denote it by $h: (G,L) \tos H$.

Observe that if $G$ or $H$ is disconnected, then each component of $G$ must be mapped to some component of $H$.
Thus, the problem can be easily reduced to the case that both $G$ and $H$ are connected. We will assume this from now on.

Recall that the non-list variant of our problem, i.e., $\lshom{H}$, is polynomial time-solvable if $H \in \cHpoly:= \{K_1, K_2, K_1^\circ\}$ (where $K_1^\circ$ denotes the one-vertex graph with a loop), and \NP-hard otherwise~\cite{DBLP:journals/tcs/FialaP05}. Let us point out that exactly the same dichotomy holds for $\llshom{H}$.

\sv{\toappendix{\section{Omitted Proofs from~\cref{sec:prel}}}}

\lv{\begin{corollary}\label{cor:poly}}
\sv{\begin{corollary}[$\spadesuit$]\label{cor:poly}}
If $H \in \cHpoly$, then $\llshom{H}$ is polynomial-time solvable, and otherwise it is \NP-hard.
\end{corollary}
\toappendix{%
  \sv{%
    \begin{corollary*}[Restated \cref{cor:poly}]
If $H \in \cHpoly$, then $\llshom{H}$ is polynomial-time solvable, and otherwise, it is \NP-hard.
\end{corollary*}
  }

\begin{proof}
As $\lshom{H}$ is a restriction of $\llshom{H}$, it is sufficient to show the polynomial cases. Indeed,
 \NP-hardness of $\lshom{H}$ implies the \NP-hardness of $\llshom{H}$.

The positive instances of $\llshom{K_1}$ are edgeless graphs, where the list of every vertex is nonempty.
Similarly, the positive instance of $\llshom{K_1^\circ}$ are graphs with no isolated vertices, where the list of every vertex is nonempty.
In both cases, \yes-instances can clearly be recognized in polynomial time.

Finally, let us consider the case that $H=K_2$, denote its vertices by $1$ and $2$.
Let $(G,L)$ be an instance of $\llshom{K_2}$. First, note that if $G$ is not bipartite, then $(G,L)$ is a no-instance, and we are done.
If $G$ is disconnected, then we can process each connected component of $G$ separately.

So assume that $G$ is connected and bipartite, and its bipartition classes are $X$ and $Y$.
We observe that every homomorphism from $G$ to $K_2$ either maps $X$ to 1 and $Y$ to 2, or maps $X$ to 2 and $Y$ to 1.
It is clear that verification if any of these two functions is a locally surjective list homomorphism can be performed in polynomial time.
\end{proof}
}

\subsection{Associated Bipartite Graphs and Associated Instances}
Now let us show that in order to classify the hard cases of \llshom{$H$} it is sufficient to consider the case that $H$ is bipartite.
A similar approach was used to the \lhom{$H$}~\cite{FullComplexity,DBLP:journals/jgt/FederHH03}, but to the best of our knowledge, we are the first to observe that it also works for \llshom{$H$}.

Let $H$ be a connected bipartite graph with bipartition classes $X,Y$, and consider an instance $(G,L)$ of $\llshom{H}$, where $G$ is connected.
Note that if $G$ is not bipartite, then $(G,L)$ is clearly a no-instance.
Thus, assume that $G$ is bipartite with the bipartition classes $A,B$.
We observe that in every homomorphism ${h: G \to H}$,
either all vertices of $A$ are mapped to $X$, and all vertices of $B$ are mapped to $Y$,
or all vertices of $A$ are mapped to $Y$, and all vertices of $B$ are mapped to $X$.
Thus in order to solve $(G,L)$, we can consider these two cases separately.
More specifically, we need to solve two instances $(G,L_1)$ and $(G,L_2)$ of $\llshom{H}$, where
\[
L_1(v) = 
\begin{cases}
L(v) \cap X & \text{ if } v \in A,\\
L(v) \cap Y & \text{ if } v \in B,
\end{cases}
\qquad\qquad
L_2(v) = 
\begin{cases}
L(v) \cap Y & \text{ if } v \in A,\\
L(v) \cap X & \text{ if } v \in B.
\end{cases}
\]
This motivates the following definition, see also~\cite{FullComplexity}.
\begin{definition}
Let $H$ be a connected bipartite graph with bipartition classes $X,Y$.
We say that an instance $(G,L)$ of $\llshom{H}$ is consistent, if
\begin{enumerate}
\item $G$ is connected bipartite with the bipartition classes $A,B$,
\item $L(A) \subseteq X$ and $L(B) \subseteq Y$.
\end{enumerate}
\end{definition}

For a graph $H=(V,E)$, by $H^* := H \times K_2$ we denote its \emph{associated bipartite graph}.
In other words, the vertex set of $H^*$ is $\{v',v'' ~:~ v \in V\}$ and the edge set is $\{u'v'' ~:~ uv \in E\}$. 
We also define $V' := \{v' ~:~v \in V\}$ and $V'' := \{v'' ~:~v \in V\}$, i.e., $V', V''$ are the bipartation classes of $H^*$.

Note that if $H$ is connected and nonbipartite, then $H^*$ is connected.
If $H$ is bipartite, then $H^*$ consists of two disjoint copies of $H$.

\lv{\begin{lemma}\label{lem:associated-equivalent}}
\sv{\begin{lemma}[$\spadesuit$]\label{lem:associated-equivalent}}
Let $H$ be a fixed connected nonbipartite graph.
Let $(G,L')$ be a consistent instance of $\llshom{H^*}$.
For each $v \in V(G)$, define $L(v) := \{x : \{x',x''\} \cap L'(v) \neq \emptyset \}$.
Then $(G,L')$ is a \yes-instance of $\llshom{H^*}$ if and only if $(G,L)$ is a \yes-instance of $\llshom{H}$.
\end{lemma}
\toappendix{%
\sv{%
  \begin{lemma*}[Restated~\cref{lem:associated-equivalent}]
  Let $H$ be a fixed connected nonbipartite graph.
  Let $(G,L')$ be a consistent instance of $\llshom{H^*}$.
  For each $v \in V(G)$, define $L(v) := \{x : \{x',x''\} \cap L'(v) \neq \emptyset \}$.
  Then $(G,L')$ is a \yes-instance of $\llshom{H^*}$ if and only if $(G,L)$ is a \yes-instance of $\llshom{H}$.
  \end{lemma*}
}
\begin{proof}
First consider $h^* : (G,L') \tos H^*$.
Define $h : V(G) \to V(H)$ as follows: $h(v) = x$ if and only if $h^*(v) \in \{x',x''\}$.
Clearly, $h$ is a homomorphism from $G$ to $H$, and it respects lists $L$.

Let us show that $h$ is locally surjective.
Consider $v \in V(G)$ such that $h(v) = x$ and some $y \in N_H(x)$.
Since $h(v) = x$, we know that $h^*(v) \in \{x',x''\}$ (the actual value depends on the bipartition class where $v$ belongs).
By symmetry, assume that $h^*(v) = x'$. Since $h^*$ is locally surjective and $x'y'' \in E(H^*)$, there is $u \in N_G(v)$, such that $h^*(u) = y''$. Then, $h(u) = y$.

Now, consider $h : (G,L) \tos H$.
Let the bipartition classes of $G$ be $A,B$, such that $L'(A) \subseteq V'$ and $L'(B) \subseteq V''$ (this holds since $(G,L')$ is consistent).
We define $h^* : V(G) \to V(H^*)$ as follows. 
Consider $v \in V(G)$ and let $h(v) = x$.
If $v \in A$, then $h^*(v) = x'$ and if $v \in B$, then $h^*(v) = x''$.
Again, it is straightforward to verify that $h^*$ is a homomorphism from $G$ to $H^*$, and it respects lists $L'$, since $(G,L')$ is consistent.

Now, let us argue that $h^*$ is locally surjective.
By symmetry, consider $v \in A$, such that $h(v) = x$. 
Then, $h^*(v) = x'$. 
Let $y'' \in N_H(x')$.
Since $h$ is locally surjective, there is $u \in N_G(v) \subseteq B$, such that $h(u) = y$.
Thus, $h^*(u) = y''$.
\end{proof}
}

\lv{\begin{corollary}\label{cor:bipartitealgotransfers}}
  \sv{\begin{corollary}[$\spadesuit$]\label{cor:bipartitealgotransfers}}
Let $\cG$ be a class of graphs and let $H$ be a fixed connected nonbipartite graph.
Suppose there is an algorithm $A$ that solves every  bipartite instance $(G,L)$ of $\llshom{H}$, such that $G \in \cG$, in time $f(|V(G)|)$.
Then there is an algorithm that solves every instance $(G,L')$ of $\llshom{H^*}$, where $G \in \cG$, in time $f(|V(G)|) \cdot |V(G)|^{\Oh(1)}$.
\end{corollary}
\toappendix{%
  \sv{%
    \begin{corollary*}[Restated~\cref{cor:bipartitealgotransfers}]
Let $\cG$ be a class of graphs and let $H$ be a fixed connected nonbipartite graph.
Suppose there is an algorithm $A$ that solves every  bipartite instance $(G,L)$ of $\llshom{H}$, such that $G \in \cG$, in time $f(|V(G)|)$.
Then there is an algorithm that solves every instance $(G,L')$ of $\llshom{H^*}$, where $G \in \cG$, in time $f(|V(G)|) \cdot |V(G)|^{\Oh(1)}$.
\end{corollary*}
}
\begin{proof}
Consider an arbitrary instance $(G,L')$ of $\llshom{H^*}$, where $G \in \cG$.
If $G$ is disconnected, we apply the reasoning below to each connected component independently.
Furthermore, we can assume that $G$ is bipartite, as otherwise, it is a trivial no-instance.

By the reasoning at the beginning of the section, we can reduce solving $(G,L')$ to solving two consistent instances.
Thus, let us assume that $(G,L')$ is consistent.
By \cref{lem:associated-equivalent}, in order to solve $(G,L')$, it is sufficient to solve the instance $(G,L)$ of $\llshom{H}$, where the lists $L$ are defined as in the lemma.
This can be done in time $f(|V(G)|)$ using the algorithm $A$.
As all additional computation (checking if $G$ is bipartite, determining lists $L$) can be performed in polynomial time, the claim follows.
\end{proof}
}

\noindent Note that \cref{cor:bipartitealgotransfers} immediately implies the following.

\begin{corollary}\label{cor:bipartitetogeneral}
Assume the ETH. Let $H$ be a fixed connected nonbipartite graph and let $\cG$ be a class of graphs.
If $\llshom{H^*}$ cannot be solved in time $2^{o(n)}$ for $n$-vertex instances in $\cG$,
then $\llshom{H}$ cannot be solved in time $2^{o(n)}$ for $n$-vertex instances in $\cG$.
\end{corollary}

\section{Algorithm for $F$-free Graphs for $F \in \cS$}\label{sec:ptalgo}

\sv{\toappendix{\section{Omitted Proof in~\cref{sec:ptalgo}}}}

An important tool used in our algorithm is the following structural result about $\{S_{t,t,t},K_3\}$-free graphs.

\begin{theorem}[Okrasa, Rz\k{a}\.zewski~\cite{DBLP:journals/corr/abs-2010-03393}]\label{thm:treewidth-sttt}
Let $t \geq 2$ be an integer.
Given an $n$-vertex $(K_3, S_{t,t,t})$-free graph $G$ with maximum degree $\Delta$,
in time $2^{\Oh(t \cdot \Delta)} \cdot n$ we can find a tree decomposition of $G$ with width at most $56t\Delta$.
\end{theorem}

\noindent Equipped by this, we are ready to prove the following algorithmic result.

\begin{theorem}\label{thm:algo-p3}
Let $a,b,c \geq 1$ be fixed integers.
The $\llshom{P_3}$ in $n$-vertex $S_{a,b,c}$-free graphs can be solved in time $2^{\Oh( (n \log n)^{2/3} )}$.
\end{theorem}
\begin{proof}
Denote the consecutive vertices of $P_3$ by $1,2,3$.
Let $(G,L)$ be an instance  of $\llshom{P_3}$, where $G$ has $n$ vertices and is $S_{a,b,c}$-free.
Let $t = \max(2,a,b,c)$ and note that $G$ is $S_{t,t,t}$-free. Furthermore, if $G$ is not bipartite, then $(G,L)$ is clearly a no-instance.
Thus, we can assume that $G$ is bipartite (and, in particular, triangle-free).

Furthermore, recall that we can safely assume that the instance $(G,L)$ is consistent.
Let $X$ and $Y$ denote the bipartition classes of $G$, such that $L(X) \subseteq \{1,3\}$ and $L(Y) = \{2\}$.
Note that if $|Y|=0$ or $|X| \leq 1$, then we are clearly dealing with a no-instance. 
Thus from now on, let us assume otherwise.
In particular, it means that every vertex from $X$ is always happy. 
Consequently, our task boils down to choosing colors for vertices of $X$ to make each vertex from $Y$ happy.

Actually, we will design a recursive algorithm that solves a slightly more general problem, where we are additionally given a function $\sigma: Y \to 2^{\{1,3\}}$.
We are looking for a list homomorphism $h : (G,L) \to P_3$, such that for every $y \in Y$ it holds that $\sigma(y) \subseteq h(N_G(y))$.
Initially, we have $\sigma(y) = \{1,3\}$ for every $y \in Y$.
Thus, the returned homomorphism is indeed locally surjective.
During the course of the algorithm, we will modify the sets $\sigma$ to keep track of the colors seen by vertices in $Y$ in the part of the graph that was removed.

Each recursive call starts with a preprocessing phase. First, we exhaustively apply the following steps.
If there is some $x \in X$ with $L(x) = \emptyset$, then we immediately terminate the recursive call and report a no-instance.
If there is some $y \in Y$ with $\sigma(y) = \emptyset$, then we can safely remove $y$ from the graph.
If there is some $x \in X$ with $|L(x)|=1$, then we remove the element of $L(x)$ from the sets $\sigma$ of all neighbors of $x$, and remove $x$ from the graph.

If none of the above steps can be applied and $G$ has an isolated vertex $y \in Y$ (note that $\sigma(y) \neq \emptyset$), then we terminate and report a no-instance. The preprocessing phase can clearly be performed in polynomial time.

Finally, if the graph obtained is disconnected, we apply the following reasoning to every connected component independently.
Let us still denote the instance by $(G,L,\sigma)$, and assume that $G$ is connected.

We consider two cases. First, suppose that there is $x \in X$ with $\deg x > (n \log n)^{1/3}$.
We branch on choosing the color for $x$, i.e., we perform two recursive calls of the algorithm, in one branch setting $L(x) = \{1\}$,
and in the other $L(x) = \{3\}$. Note that at least $\deg x/3 \geq (n \log n)^{1/3} /3$ neighbors of $x$ have the same set $\sigma$.
Consequently, in at least one branch, the sets $\sigma$ will be reduced for at least $(n \log n)^{1/3} /3$ vertices during the preprocessing phase. In the other branch, we are guaranteed to have a little progress, too: the vertex $x$ will be removed from the graph.
Let us define the measure $\mu$ of the instance as $\mu := \sum_{x \in X} |L(x)| + \sum_{y \in Y} |\sigma(y)|$.
Clearly $n \leq \mu \leq 2n$.
Thus the complexity of this step is given by the following recursive inequality:
\[
F(\mu) \leq F(\mu-2) + F(\mu - (n \log n)^{1/3}/3) = \mu^{\Oh(\mu / (n \log n)^{1/3})} = 2^{\Oh( (n \log n)^{2/3})}.
\]

So now let us assume that for each $x \in X$ it holds that $\deg x < (n \log n)^{1/3}$.
Let $Y'$ be the set of vertices $y \in Y$ satisfying $\deg y \geq (n \log n)^{2/3}$.
Observe that $|E(G)| \leq |X| \cdot (n \log n)^{1/3} \leq n^{4/3} \log^{1/3}n$.
Consequently, $|Y'| \leq  |E(G)| / (n \log n)^{2/3} \\ \leq n^{2/3}$. 

Consider the graph $G' := G - Y'$. As it is an induced subgraph of $G$, it is $(K_3, S_{t,t,t})$-free.
Furthermore, the maximum degree of $G'$ is at most $(n \log n)^{2/3}$. 
Consequently by \cref{thm:treewidth-sttt},
in time $2^{\Oh( (n \log n)^{2/3})}$ we can find a tree decomposition of $G'$ with width $\Oh( (n \log n)^{2/3})$.
Let us modify this tree decomposition by adding the set $Y'$ to every bag -- this way we obtain a tree decomposition of $G$ with width $\Oh( (n \log n)^{2/3} + n^{2/3})=\Oh( (n \log n)^{2/3})$.

Using fairly standard dynamic programming on a tree decomposition, we can solve our auxiliary problem on graphs given with a tree decomposition of width $w$ in time $2^{\Oh(w)} \cdot n^{\Oh(1)}$. 
Indeed, the state of the dynamic programming is the coloring of the vertices from $X$ in the current bag and the colors seen by the vertices from $Y$ in subgraph induced the subtree rooted at the current bag (these colors are reflected in sets $\sigma$). 
Thus, the total number of states to consider is at most $3^w$ (two possibilities for a vertex from $X$ and at most three for a vertex from $Y$).

Consequently, in the second case we obtain the running time $2^{\Oh( (n \log n)^{2/3})}+2^{\Oh( (n \log n)^{2/3})} = 2^{\Oh( (n \log n)^{2/3})}$.

Summing up, the overall complexity of the algorithm is $2^{\Oh( (n \log n)^{2/3})}$. This completes the proof.
\end{proof}

Note that every $P_t$-free graph is also, e.g., $S_{t,1,1,}$-free, so \cref{thm:algo-p3} can also be applied to $P_t$-free graphs.
Now, let us show a slight generalization of \cref{thm:algo-p3} to the case that we exclude a forest of paths and subdivided claws.

\lv{\begin{theorem}\label{thm:algo-p3-forests}}
  \sv{\begin{theorem}[$\spadesuit$]\label{thm:algo-p3-forests}}
For every $F \in \cS$, the $\llshom{P_3}$ problem in $n$-vertex $F$-free graphs can be solved in time $2^{\Oh(( n \log n)^{2/3})}$.
\end{theorem}
\toappendix{%
  \sv{%
    \begin{theorem*}[Restated \cref{thm:algo-p3-forests}]
For every $F \in \cS$, the $\llshom{P_3}$ problem in $n$-vertex $F$-free graphs can be solved in time $2^{\Oh(( n \log n)^{2/3})}$.
\end{theorem*}
}

\begin{proof}[Sketch of proof]
Let $F = F_1 + F_2 + \ldots + F_p$ for some $p \geq 1$, where each $F_i$ is a subdivided claw.

We begin similarly to the algorithm from \cref{thm:algo-p3}. 
Again, we are solving an auxiliary problem with instance $(G,L,\sigma)$.
First, we check if the instance graph is bipartite, and otherwise, we reject it. 
Let the bipartition classes of $G$ be $X$ and $Y$, and let $L(X) \subseteq \{1,3\}$ and $L(Y)=\{2\}$.
Then, we perform the preprocessing phase and the branching phase; note that in these phases, we do not assume anything about the forbidden induced graph.
The recursion tree has $2^{\Oh((n \log n^{2/3})}$ leaves, each corresponds to an instance which is $(K_3,F)$-free and every vertex from $X$ has maximum degree at most $ (n \log n)^{1/3}$. Consider one such instance, for simplicity let us call it $(G,L)$.

We continue as in the proof of \cref{thm:algo-p3} by selecting the set $Y' \subseteq Y$ of vertices of degree at least $(n \log n)^{2/3}$.
Recall that $|Y'| \leq n^{2/3}$ and the graph $G' - Y'$ is of maximum degree at most $( n \log n)^{2/3}$.

Now for each $i=1,\ldots,p-1$ we perform the following steps.
Let $(G',L')$ be an instance corresponding to a leaf of the recursion tree, with the set $Y'$ removed.
We check if $G'$ contains $F_i$ as an induced subgraph, this can be done in polynomial time by the exhaustive enumeration. 
If not, then $G'$ is $(K_3,F_i)$-free and we can call the algorithm given by \cref{thm:treewidth-sttt} and continue exactly as in the proof of \cref{thm:algo-p3}.

Thus, let us suppose that there is $S \subseteq V(G')$, such that $G'[S] \simeq F_i$.
We observe that $|N[S]| = \Oh( (n \log n)^{2/3})$.
We exhaustively guess the coloring of $N[S] \cap X$, this results in $2^{\Oh( (n\log n)^{1/3})}$ branches.
In each branch, we update the sets $\sigma$ for the neighbors of colored vertices; in particular we reject if some vertex from $y \in S \cap Y$ does not see some color in $\sigma(y)$. Note that each instance is $(K_3,F_{i+1}+\ldots,+F_p)$-free.

After the last iteration, the instances corresponding to the leaves of the recursion tree are $(K_3,F_p)$-free, and thus we continue as in the proof of \cref{thm:algo-p3}, i.e., use \cref{thm:treewidth-sttt}, restore the set $Y'$, and solve the problem by dynamic programming.

The total number of leaves of the recursion tree is at most 
\[ \underbrace{2^{\Oh((n \log n^{2/3})}}_{\substack{\text{branching on}\\\text{a high-degree vertex in $X$}}} \cdot \prod_{i=1}^{p-1} \underbrace{2^{\Oh((n \log n^{2/3})}}_{\substack{\text{branching on the neighborhood}\\ \text{of an induced copy of $F_i$}}} = 2^{\Oh((n \log n^{2/3})},
\] each of which corresponds to an instance that is $(K_3,F_i)$-free for some $i \in [p]$, and thus can be solved in time $2^{\Oh((n \log n^{2/3})}$ as in \cref{thm:algo-p3}.
\end{proof}
}

Now, let us show that the algorithm from \cref{thm:algo-p3} can be used to solve $\llshom{C_4}$ in $P_t$-free graphs.
The proof of the following lemma is based on a similar argument used by Okrasa and Rz\k{a}\.zewski~\cite{DBLP:journals/jcss/OkrasaR20} in the non-list case.

\begin{lemma} \label{lem:c4top3}
Let $(G,L)$ be a consistent instance of $\llshom{C_4}$.
Then the problem can be reduced in polynomial time to solving two consistent instances of $\llshom{P_3}$.
\end{lemma}
\begin{proof}
Let the consecutive vertices of $C_4$ be $1, 2, 3, 4$.
Let the bipartition classes of $G$ be $X,Y$, and let $L(X) \subseteq \{1,3\}$ and $L(Y) \subseteq \{2, 4\}$.

Define $L',L''$ as follows:
\begin{align*}
L'(v) = \begin{cases}
L(v) & \text{ if } v \in X\\
\{2\} & \text{ if } v \in Y,
\end{cases} & \quad \text{ and } \quad 
L''(v) = \begin{cases}
\{1\} & \text{ if } v \in X\\
L(v) & \text{ if } v \in Y.
\end{cases}
\end{align*}
We claim that $(G,L)$ is a \yes-instance of $\llshom{C_4}$ if and only if $(G,L')$ is a \yes-instance of $\llshom{C_4[1,2,3]}$ and $(G,L'')$ is a \yes-instance of $\llshom{C_4[4,1,2]}$. 
Note that both graphs $C_4[1,2,3]$ and $C_4[4,1,2]$ are induced three-vertex paths.

First, suppose that there is some $h : (G,L) \tos C_4$.
Let us define $h',h'' : V(G) \to \{1,2,3,4\}$ as follows:
\begin{align*}
h'(v) = \begin{cases}
h(v) & \text{ if } v \in X\\
2 & \text{ if } v \in Y,
\end{cases} & \quad \text{ and } \quad 
h''(v) = \begin{cases}
1 & \text{ if } v \in X\\
h(v) & \text{ if } v \in Y.
\end{cases}
\end{align*}
Let us argue that that $h': (G,L') \tos C_4[1,2,3]$ and $h'' : (G,L'') \tos C_4[4,1,2]$.
We prove only the first claim. The proof of the second one is analogous.
First, $h'$ is clearly a homomorphism. 
Furthermore, it satisfies the lists, as for $x \in X$ we have $h'(x) = h(x) \in L'(x) = L(x)$,
and for $y \in Y$ we have $h'(y) = 2 \in \{2\} = L'(y)$.
Finally, let us argue that $h'$ is locally surjective.
For each $x \in X$, there are some $y,y'$ such that $h(y)=2$ and $h(y')=4$, as otherwise $h$ is not locally surjective. 
Thus, $h'(y)=h'(y')=2$, which makes $x$ happy in $h'$.
Similarly, for each $y \in Y$, there are some $x,x'$ such that $h(x)=1$ and $h(x')=3$, as otherwise $h$ is not locally surjective. 
Thus, $h'(x)=1$ and $h'(x')=3$, which makes $y$ happy in $h'$.

Now, suppose there are  $h': (G,L') \tos C_4[1,2,3]$ and $h'' : (G,L'') \tos C_4[4,1,2]$.
We define $h : V(G) \to \{1,2,3,4\}$ as follows:
\[
h(v) = \begin{cases}
h'(v) & \text{ if } v \in X\\
h''(v) & \text{ if } v \in Y.
\end{cases}
\]
Let us argue that $h: (G,L) \tos C_4$.

First, note that $h$ is a homomorphism, as for every edge $xy \in E(G)$, where $x \in X$ and $y \in Y$,
we have $h(x) \in \{1,3\}$ and $h(y) \in \{2,4\}$.

Now, observe that $h$ respects the lists $L$.
Indeed, for $x \in X$ we have $h(x) = h'(x) \in L'(x) = L(x)$ and for $y \in Y$ we have $h(y) = h''(y) \in L''(y) = L(y)$.

Finally, let us argue that $h$ is locally surjective. 
Consider some $x \in X$; the argument for vertices in $Y$ is symmetric.
Note that $h''(x) = 1$. 
Since $h''$ is locally surjective, $x$ has two neighbors $y,y'$, such that $h''(y)=2$ and $h''(y')=4$.
Consequently, $h(y)=2$ and $h(y')=4$, which makes $x$ happy in $h$.
This completes the proof.
\end{proof}

\noindent Combining \cref{lem:c4top3} with \cref{thm:algo-p3-forests}, we immediately obtain the following.

\begin{theorem}\label{thm:algo-c4-forests}
For every $F \in \cS$, the $\llshom{C_4}$ in $n$-vertex $F$-free graphs can be solved in time $2^{\Oh(( n \log n)^{2/3})}$.
\end{theorem}

\section{Hardness for $F$-free Graphs for $F \in \cS$}\label{sec:pthard}

\sv{\toappendix{\section{Omitted Parts of~\cref{sec:pthard}}\label{app:pt}}}

In this section, we will prove the hardness part of \cref{thm:pathfree}, i.e., if ${H \not \in \cHpoly \cup \{P_3,C_4\}}$ then there is $t$ such that \llshom{$H$} cannot be solved in subexponential time in $P_t$-free graphs.
It easily follows that such $H$ contains at least one of $K_2^\circ, K_2^{\circ\circ}, K_3, P_4, K_{1,3}$ as an induced subgraph, where $K^\circ_2$ and $K^{\circ\circ}_2$ are graphs consisting of an edge with a loop at one or both of its endpoint, respectively.
First, we will prove the theorem for several base cases for $H \in \{K_{1,3}, P_4, K^{\circ\circ}_2\}$.
\sv{%
  \begin{theorem}[$\spadesuit$]\label{thm:universal}
For $H\in\{K_{1,3}, P_4, K^{\circ\circ}_2\}$ the \llshom{$H$} problem cannot be solved in time $2^{o(n)}$ in $n$-vertex $P_{14}$-free graphs, unless the ETH fails.
Moreover, the problem is hard even for instances where all the lists are of size at most 2, and each vertex with a list of size exactly two has a neighbor with a list of size exactly one.
\end{theorem}
The proof follows by Theorems~\ref{thm:K13}, \ref{thm:P4}, and~\ref{thm:K2} in \cref{app:pt}.
}
Then, we will generalize the result for $H$ containing an induced subgraph $H' \in \{K_{1,3}, P_4, K^{\circ\circ}_2\}$.
The last cases when $H$ contains $K^\circ_2$ or $K_3$ as an induced subgraph will follows from the base cases and \cref{cor:bipartitetogeneral} as for such $H$ the graph $H^*$ contains $P_4$ as an induced subgraph.

\toappendix{%
\subsection{Hardness for $H = K_{1,3}$}

\begin{theorem}\label{thm:K13}
The \llshom{$K_{1,3}$} problem cannot be solved in time $2^{o(n)}$ in $n$-vertex $P_{10}$-free graphs, unless the ETH fails.
Moreover, the constructed reduction use only lists of size at most 2, and each vertex of the construction which has a list of size exactly two has a neighbor with a list of size exactly one.
\end{theorem}
\begin{proof}
Let $H=K_{1,3}$, where $0$ is the central degree three vertex and $1,2,3$ are the leaves.
We reduce from the 3-\textsc{SAT} problem. The ETH implies that an instance with $N$ variables and $M$ clauses cannot be solved in time $2^{o(N+M)}$~\cite{DBLP:books/sp/CyganFKLMPPS15}.
We first describe the reduction.
For every variable $z$, a variable gadget $\vrb(z)$ is represented by $K_{1,3}$ with the central vertex $v_0$ with $L(v_0)=\{0\}$ and the leaves $x,x',w$ with lists $\{1,2\}$, $\{1,2\}$, and $\{3\}$, respectively.

\begin{claimm}[Variable consistency]\label{cl:K13:var}
There exists exactly two homomorphisms $h_1,h_2:(\vrb(z),L) \to H$ in which vertices of $\vrb(z)$ are happy:  
\begin{enumerate}
\item[(i)] $h_1(x)=1$, $h_1(x')=2$, $h_1(v_0)=0$, $h_1(w)=3$, and \label{cl:item:true} 
\item[(ii)] $h_2(x)=2$, $h_2(x')=1$, $h_2(v_0)=0$, $h_2(w)=3$. \label{cl:item:false} 
\end{enumerate}
\end{claimm}

\noindent\textit{Proof of Claim~\ref{cl:K13:var}.}~
In order to make vertex $v_0$ happy, vertices $x$ and $x'$ have to be mapped to different vertices of $H$.
It is easy to check that all other vertices are happy.
  \hfill $\diamondsuit$
  \medskip

For every clause $c$, the clause gadget $\cls(c)$ is represented by $S_{3,3,3}$. 
Let $y$ denote the central vertex, and $s_1^i$, $s_2^i$, and $s_3^i$ denote vertices of the $i$-th branch ordered by the distance from $y$, for $i\in\{1,2,3\}$.
We set the lists as $L(y)=\{0\}$, $L(s_1^i)=\{1,3\}$, $L(s_2^i)=\{0\}$, and $L(s_3^i)=\{2,3\}$.
We construct the final graph $G$ as follows.
If variable $z$ appears as the $i$-th literal in clause $c$, then 
$s_2^i$ (of $\cls(c)$) is adjacent to $x$ or $x'$ (of $\vrb(z)$) if the literal is positive or negative, respectively.
Moreover, each central vertex $y$ of each clause is connected with each $x$ and with each $x'$ vertices of all variable gadgets.
For an overview of the described construction, see Figure~\ref{fig:K13}.
This completes the construction of the graph $G$.
The described construction gives only a constant blow-up, and therefore the ETH lower bound is preserved.
Observe that vertex $y$ is happy if and only if for at least one $i\in\{1,2,3\}$ has $h(s_1^i)=3$.

\begin{figure}[h]
\centering
\includegraphics[scale=0.8]{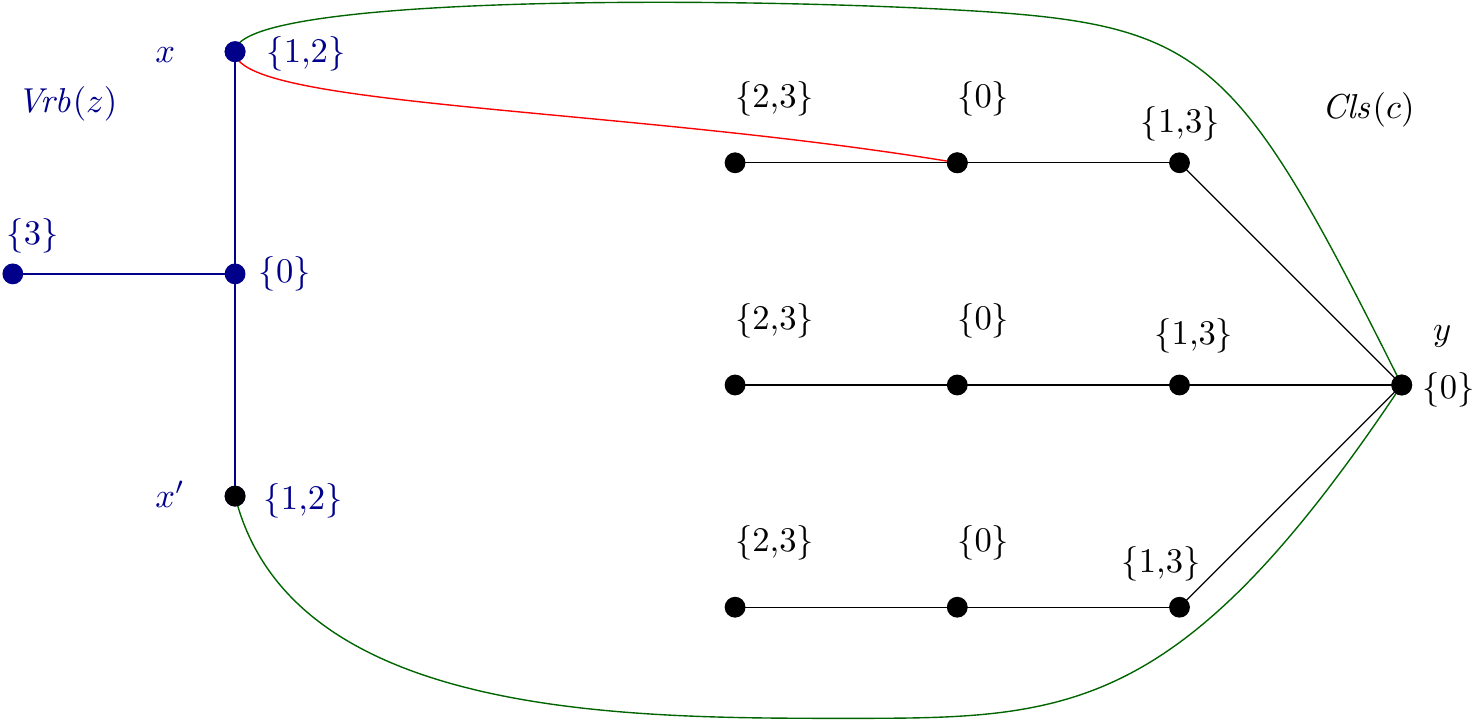} 
\caption{An example of variable gadgets connected to a clause gadget for the case $H=K_{1,3}$.
  The variable gadget is depicted by blue color, clause gadget by black color, green are connections between each variable and clause gadgets, and red is the connection representing an occurrence of a variable in the clause (in this example, a positive occurrence).
}
\label{fig:K13}
\end{figure}

First, suppose that we have a satisfying assignment $\phi$ of the formula. 
We will construct the sought homomorphism $h: (G,L) \tos K_{1,3}$.
For each variable $z$, we set $h$ of $\vrb(z)$ based on Claim~\ref{cl:K13:var}: we use assignment from (i) if $\phi(z)$ is true and (ii) otherwise.
As each clause $c$ is satisfied, say by the $i$-th literal, there is a corresponding variable gadget such that $s_2^i$ sees a vertex mapped to 1 (either a true variable occurring as a positive literal or a false variable occurring as a negative literal in $c$).
In order to make $s_2^i$ happy, we set $h(s_3^i)=2$ and so $h(s_1^i)=3$, that immediately makes $y$ also happy.
It is easy to check that all other vertices in a clause gadget may be mapped in a way to be happy and that we created a $\llshom{K_{1,3}}$ \yes-instance.

Second, suppose  there is a homomorphism $h: (G,L) \tos K_{1,3}$.
We define the truthfulness of variables based on Claim~\ref{cl:K13:var}: $x$ is true if $h(x) = 1$ (i) and false if $h(x) = 2$ (ii).
Now, it remains to check whether this is a satisfying assignment.
For each clause $c$, $y$ in $\cls(c)$ is happy. Therefore for some $i$, $h(s_1^i)=3$.
In order for $s_2^i$ to be happy, vertex $x$ (or $x'$) in the adjacent variable gadget $\vrb(z)$ has to be mapped to $1$. 
In both cases, variable $z$ satisfies $c$.

It remains to argue that the constructed graph is $P_{10}$-free.
Suppose there is an induced $P_{10}$ denoted as $P$.
Clearly, it is not solely within variable or clause gadgets.
Hence, it must contain vertex $x$ or $x'$ of some variable gadget as those are the only vertices that connect variable and clause gadgets.
We denote this particular vertex $p$.
If no central vertex $y$ of a clause gadget belongs to $P$, we are limited to an induced path on at most seven vertices.
Therefore, $P$ contains a central vertex of some gadget.  
As all central vertices of clause gadgets induce together with all vertices $x$ and $x'$ of variable gadgets a complete bipartite graph, either at most one other central vertex or at most one other vertex $x$ or $x'$ can belong to $P$. The former leads to an induced path on at most seven vertices, the latter to an induced path on at most nice vertices.
\end{proof}

\subsection{Hardness for $H = P_4$}

\begin{theorem}\label{thm:P4}
  The \llshom{$P_4$} problem cannot be solved in time $2^{o(n)}$ in $n$-vertex $P_{14}$-free graphs, unless the ETH fails.
Moreover, the problem is hard even for instances where all the lists are of size at most 2, and each vertex that has a list of size exactly two has a neighbor with a list of size exactly one.
\end{theorem}
\begin{proof}
Let $H=P_4$, where $1,2,3,4$ are the consecutive vertices on the path.
We reduce from the 3-\textsc{SAT} problem.
Let $\Phi$ be an instance of 3-\textsc{SAT}. %
Without loss of generality, we can assume that each variable appears at least once negatively and at least once positively in some clause in $\Phi$.
For every variable $x$ we introduce a variable gadget $\vrb(x)$ to be a double-sided broom with a five-vertex path $r_1 r_2 r_3 r_4 r_5$, called \emph{handle}, set $X$ of leaves on one end, and set $X'$ of leaves on the second end of the handle.
The number of vertices in $X$ and $X'$ equals to the number of positive and negative occurrences of $x$ in $\Phi$, respectively.
We set the lists $L$ of $r_1, r_2, r_3, r_4, r_5$ to be $\{1,3\},\{2,4\},\{3\},\{2,4\},\{1,3\}$, respectively, and the lists of all vertices in $X\cup X'$ to be $\{2\}$ (see a part of Figure~\ref{fig:P4} for an illustration).

\begin{claimm}[Variable consistency]\label{cl:P4:var}
There exists exactly two homomorphisms $h_1,h_2:(\vrb(x),L) \to H$ in which the vertices of the handle of $\vrb(x)$ are happy:  
\begin{enumerate}
\item[(i)] $h_1(r_1)=1$, $h_1(r_2)=2$, $h_1(r_3)=3$, $h_1(r_4)=4$, $h_1(r_5)=3$, $h_1(u)=2, u\in X\cup X'$, and \label{cl:item:true} 
\item[(ii)] $h_2(r_1)=3$, $h_2(r_2)=4$, $h_2(r_3)=3$, $h_2(r_4)=2$, $h_2(r_5)=1$, $h_2(u)=2, u\in X\cup X'$.\label{cl:item:false} 
\end{enumerate}
\end{claimm}

\noindent\textit{Proof of Claim~\ref{cl:P4:var}.}~
In order to make vertex $r_3$ happy, vertices $r_2$ and $r_4$ have to be mapped to different vertices of $H$.
As consequence, vertices $r_1,r_5$ are mapped to different vertices as well.
 It is easy to check that all other vertices except $X\cup X'$ are happy.
  \hfill $\diamondsuit$
  \medskip

\begin{figure}[h]
\centering
\includegraphics[scale=0.8]{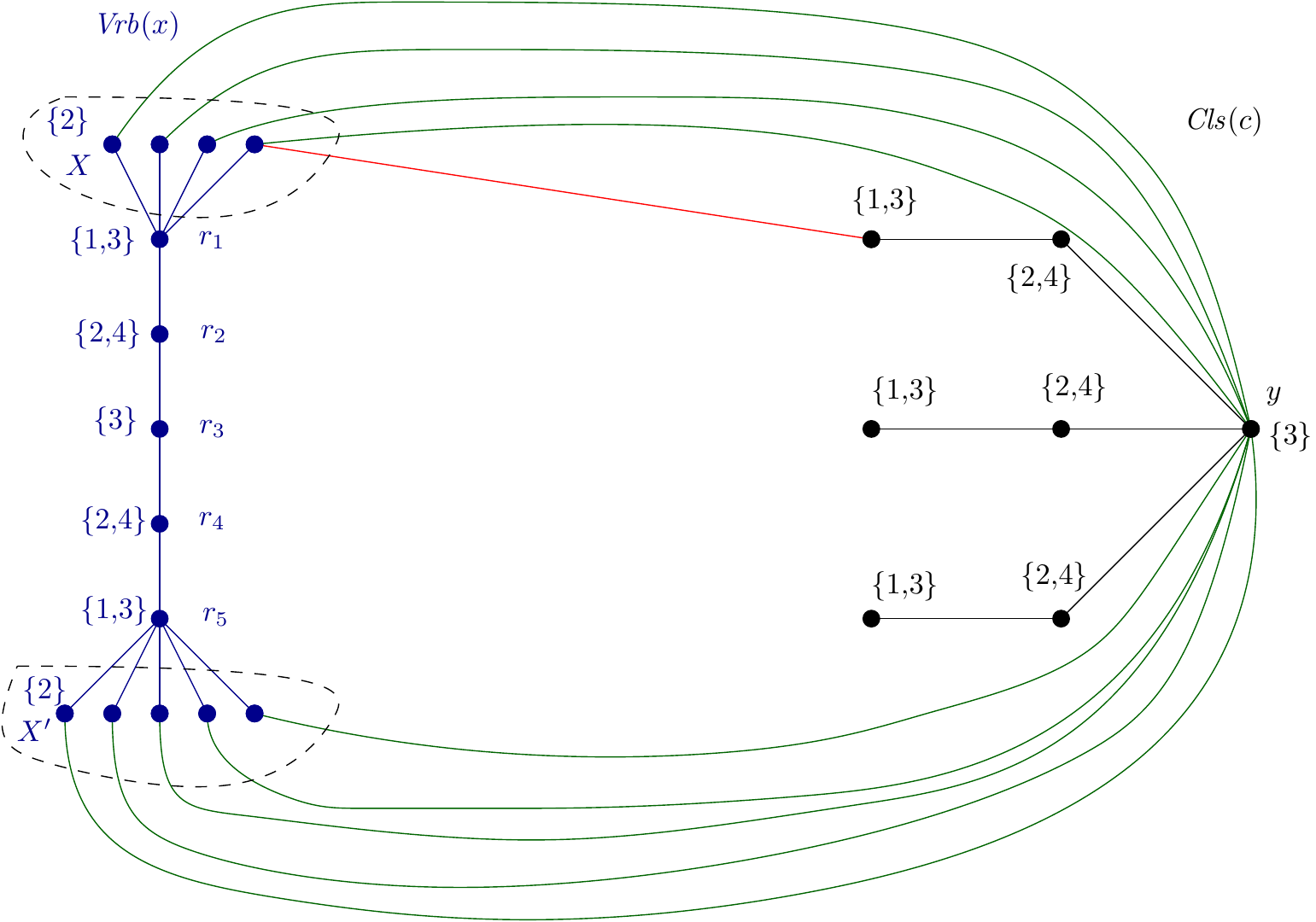} 
\caption{An example of variable gadgets connected to a clause gadget for the case ${H=P_4}$.
  The variable gadget is depicted by blue color, clause gadget by black color, green are connections between each variable and clause gadgets, and red is the connection representing an occurrence of a variable in the clause (in this example, a positive occurrence).
}
\label{fig:P4}
\end{figure}

Now, for every clause $c$, we create a subdivided claw $S_{2,2,2}$ as a clause gadget $\cls(c)$.
Let $y$ denote its central vertex and $s_i^j$ denote the vertex in distance $i$ from $y$ in the $j$-th branch, $i\in\{1,2\}, j\in\{1,2,3\}$.
We set the lists as $L(y)=\{3\}$, $L(s_1^j)=\{2,4\}$, $L(s_2^j)=\{1,3\}$.

We construct the final graph as follows.
If a variable $x$ occurs as the $j$-th literal in clause $c$, we connect $s_2^j$ (in $\cls(c)$)  with one vertex from $X$, or $X'$ (in $\vrb(x)$) if the literal is possitive or negative, respectively.
Every vertex from $X\cup X'$ is connected to exactly one $s_2^j$ vertex in some clause gadget.
Moreover, each vertex $y$ of each clause is adjacent to all vertices in $X$ and $X'$ of all variable gadgets.
For an overview of the described construction, see Figure~\ref{fig:P4}.
This completes the construction of the graph $G$.
The described construction gives only a constant blow-up, and therefore the ETH lower bound is preserved.
Observe that vertex $y$ is happy if and only if for at least one $i\in\{1,2,3\}$ has $h(s_1^i)=4$.

First, suppose that we have a satisfying assignment $\phi$ of the formula $\Phi$, and we will construct the sought homomorphism $h: (G,L) \tos P_4$.
For each variable $x$, we set $h$ of $\vrb(x)$ based on Claim~\ref{cl:P4:var}: we use assignment from (i) if $\phi(x)$ is true and (ii) otherwise.
As each clause is satisfied, there is an adjacent true variable gadget connected as positive literal or false variable gadget connected as negative literal.
As those cases are symmetric, we infer that the appropriate $s_2^j$ on the $j$-th branch that connects the clause gadget with the respective variable gadget is allowed to be mapped to $3$ as $h(r_1)=1$ (or $h(r_5)=1$ in case of a false variable with negative occurrence).
Therefore, we set $h(s_1^i)=4$, making the corresponding $y$ happy. 
It is easy to check that all other vertices in a clause gadget are happy and that we created a $\llshom{P_4}$ \yes-instance.

Second, suppose  there is a homomorphism $h: (G,L) \tos P_4$.
We define the truthfulness of variables based on Claim~\ref{cl:P4:var}: $x$ is true if $h(r_1)=1$ (i) and false if $h(r_1)=3$ (ii).
Now, it remains to check whether this is a satisfying assignment.
For each clause $c$, the center vertex $y$ of $\cls(c)$  is happy and therefore, as observed, for some $j$ $h(s_1^j)=4$.
Hence, $h(s_2^j)=3$. 
In order for the corresponding vertex in $X$ (or $X'$) to be happy, $r_1$ (or $r_5$) has to be mapped to $1$.
That, in both cases, is the variable that satisfies the particular clause (either a true variable with positive occurrence or a false variable with negative occurrence).

It remains to argue that the constructed graph is $P_{14}$-free.
Suppose there is an induced $P_{14}$ denoted as $P$.
Clearly, it is not solely within variable or clause gadgets.
Hence, it must contain a vertex $p\in X\cup X'$ of some variable gadget as those are the only vertices that connect variable and clause gadgets.
We distinguish two cases:
\begin{enumerate}
      \item There exists a central vertex $y$ of a $\cls(c)$ that belongs to $P$.
        Now, $P$ may contain one other vertex in $X$ (or $X'$) from a different (or possibly the same) variable gadget, but not more.
    In this case, we construct a path on at most $13$ vertices. 
    Or it may contain another central vertex from a different clause, but no more.
In this case, we may construct a path on at most $7$ vertices.
  \item No central vertex of any clause is in $P$.
    In this case, we are limited to an induced path on at most $11$ vertices.\qedhere
  \end{enumerate}
\end{proof}

\subsection{Hardness for $H = K_2^{\circ\circ}$}

\begin{theorem}\label{thm:K2}
  The \llshom{$K_2^{\circ\circ}$} problem cannot be solved in time $2^{o(n)}$ in $n$-vertex $P_{12}$-free graphs, unless the ETH fails.
Moreover, the problem is hard even for instances where all the lists are of size at most 2, and each vertex that has a list of size exactly two has a neighbor with a list of size exactly one.
\end{theorem}

\begin{proof}
Let $H=K_2^{\circ\circ}$, with vertices $1,2$.
We reduce from the 3-\textsc{SAT} problem.
Let $\Phi$ be an instance of 3-\textsc{SAT}. %
Without loss of generality, we can assume that each variable appears at least once negatively and at least once positively in some clause of $\Phi$.
For every variable $x$, we introduce a variable gadget $\vrb(x)$ to be a double-sided broom with a three-vertex path $r_1 r_2 r_3$, called \emph{handle}, set $X$ of leaves on one end, and set $X'$ of leaves on the second end of the handle.
The number of vertices in $X$ and $X'$ equals to the number of positive and negative occurrences of $x$ in $\Phi$, respectively.
We set the lists $L$ of $r_1, r_2, r_3$ to be  $\{1,2\},\{1\},\{1,2\}$, respectively, and the lists of all vertices in $X\cup X'$ to be $\{2\}$ (see a part of Figure~\ref{fig:K2}).

\begin{claimm}[Variable consistency]\label{cl:K2:var}
There exists exactly two homomorphisms $h_1,h_2:(\vrb(x),L) \to H$ in which the vertices of the handle of $\vrb(x)$ are happy:  
\begin{enumerate}
\item[(i)] $h_1(r_1)=1$, $h_1(r_2)=1$, $h_1(r_3)=2$, $h_1(u)=2, u\in X\cup X'$, and 
\item[(ii)] $h_2(r_1)=2$, $h_2(r_2)=1$, $h_2(r_3)=1$, $h_2(u)=2, u\in X\cup X'$.
\end{enumerate}
\end{claimm}

\noindent\textit{Proof of~\cref{cl:K2:var}.}~
In order to make vertex $r_2$ happy, vertices $r_1$ and $r_3$ have to be mapped to different vertices of $H$.
 It is easy to check that all other vertices except $X\cup X'$ are happy.
  \hfill $\diamondsuit$
\medskip
\begin{figure}[h]
\centering
\includegraphics[scale=0.8]{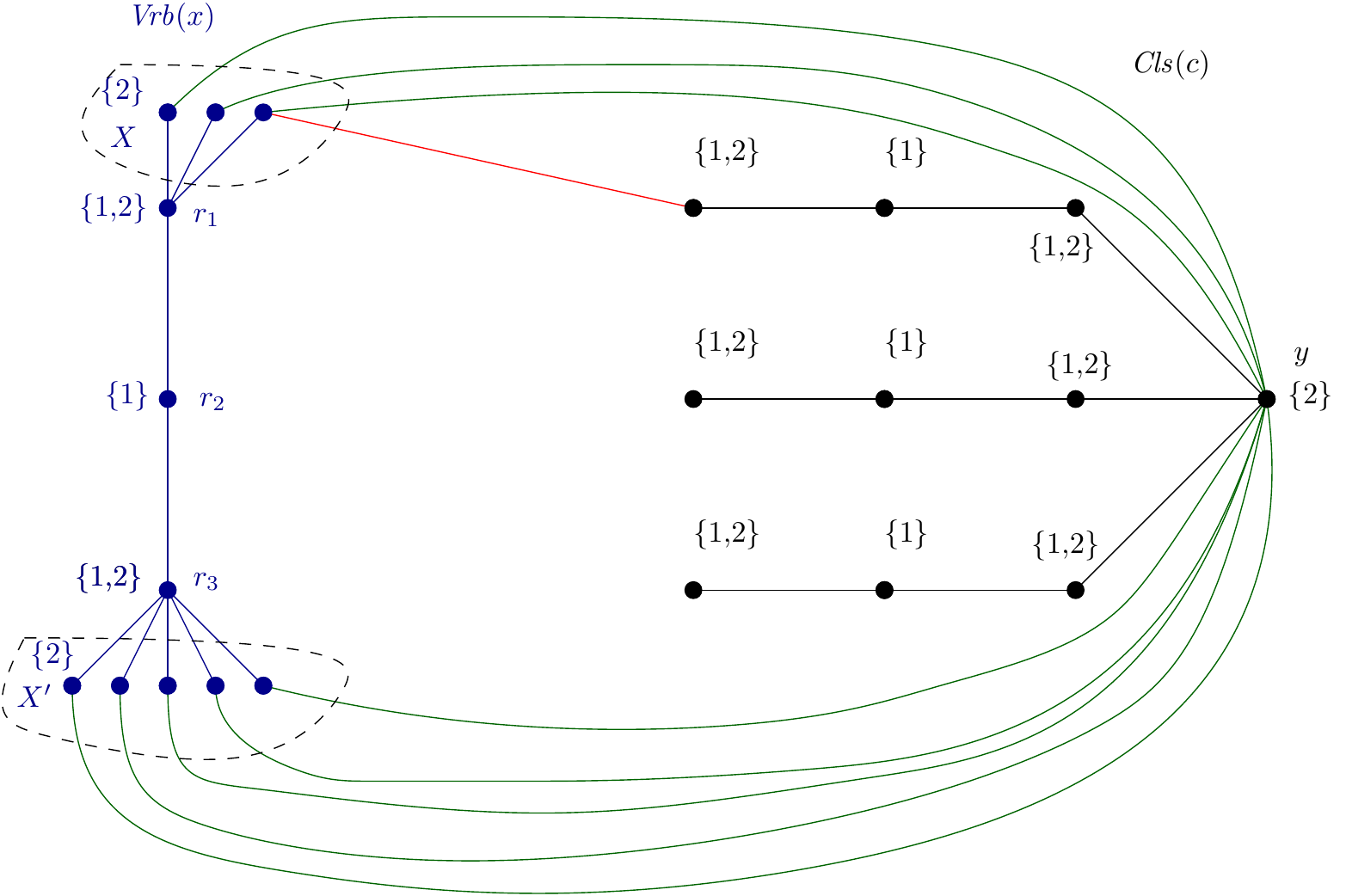} 
\caption{An example of variable gadgets connected to a clause gadget for the case ${H=K_2^{\circ\circ}}$.
  The variable gadget is depicted by blue color, clause gadget by black color, green are connections between each variable and clause gadgets, and red is the connection representing an occurrence of a variable in the clause (in this example, a positive occurrence).
}
\label{fig:K2}
\end{figure}

Now, for every clause $c$, we create a subdivided claw $S_{3,3,3}$ as a clause gadget $\cls(c)$. 
We denote by $y$ the central vertex and by $s_i^j$ the vertex in distance $i$ from $y$ in the $j$-th branch, $i\in\{1,2,3\}, j\in\{1,2,3\}.$
We set the lists as $L(y)=\{2\}$, $L(s_1^j)=\{1,2\}$, $L(s_2^j)=\{1\}$, and $L(s_3^j)=\{1,2\}$.

We construct the final graph as follows.
If a variable $x$ occurs as the $j$-th literal in clause $c$, we connect $s_3^j$ (in $\cls(c)$)  with one vertex from $X$, or $X'$ (in $\vrb(x)$) if the literal is possitive or negative, respectively.
We maintain the property that every vertex from $X\cup X'$ is connected to exactly one $s_3^j$ vertex in some clause gadget.
Moreover, each vertex $y$ of each clause is adjacent to all vertices in $X$ and $X'$ of all variable gadgets.
For an overview of the described construction, see \cref{fig:K2}.
This completes the construction of the graph $G$.
  The described construction gives only a constant blow-up, and therefore the ETH lower bound is preserved.
  Observe that vertex $y$ is happy if and only if for at least one $i\in\{1,2,3\}$ has $h(s_1^i)=1$.

First, suppose that we have a satisfying assignment $\phi$ of the formula $\Phi$, and we construct the sought homomorphism $h$.
For each variable $x$, we set $h$ of $\vrb(x)$ based on Claim~\ref{cl:K2:var}: we use the assignment from (i) if $\phi(x)$ is true and from (ii) otherwise.
As each clause is satisfied, there is an adjacent true variable gadget connected as positive literal or false variable gadget connected as negative literal.
As those cases are symmetric, we infer that the appropriate $s_3^j$ on the $j$-th branch that connects the clause gadget with the respective variable gadget is allowed to be mapped to $2$ as $h(r_1)=1$ (or $h(r_3)=1$ in case of a false variable with negative occurrence).
Therefore, in order to make $s_2^i$ happy, we set $h(s_1^i)=1$, which makes $y$ happy as well. 
It is easy to check that all other vertices in a clause gadget are happy and that we created a $\llshom{K_2^{\circ\circ}}$ \yes-instance.

Second, suppose  there is a homomorphism $h: (G,L) \tos K_2^{\circ\circ}$.
We define the truthfulness of variables based on Claim~\ref{cl:K2:var}: $x$ is true if $h(r_1)=1$ (i) and false if $h(r_1)=2$ (ii).

Now, it remains to check whether this is a satisfying assignment.
As observed, for each clause $c$, the center vertex $y$ of $\cls(c)$ is happy and therefore for some $j$ $h(s_1^j)=1$.
Hence, $h(s_3^j)=2$ as $s_2^j$ is happy. 
In order for the corresponding vertex in $X$ (or $X'$) to be happy, $r_1$ (or $r_3$) is  mapped to $1$.
That, in both cases, is the variable that satisfies the particular clause (either a true variable with a positive occurrence or a false variable with a negative occurrence).

It remains to argue that the constructed graph is $P_{12}$-free.
Suppose there is an induced $P_{12}$ denoted as $P$.
Clearly, it is not solely within variable or clause gadgets.
Hence, it must contain a vertex $p\in X\cup X'$ of some variable as those are the only vertices that connect variable and clause gadgets.
We distinguish two cases: 
\begin{enumerate}
      \item There exists a central vertex $y$ of a $\cls(c)$ that belongs to $P$.
        Now, $P$ may contain one other vertex in $X$ (or in $X'$) from a different (or possibly the same) variable gadget, but not more.
    In this case, we construct a path on at most $9$ vertices. 
    Or it also may contain another central vertex from a different clause, but no more.
In this case, we may construct a path on at most $9$ vertices.
  \item No central vertex of any clause is in $P$.
    In this case, we are limited to an induced path on at most $11$ vertices.\qedhere
  \end{enumerate}
\end{proof}

Let us point out that $(K_2^{\circ\circ})^*=C_4$, and for every $t$, the \llshom{$C_4$} problem is subexponential-time solvable in $P_t$-free graphs by \cref{thm:algo-c4-forests}.
Thus the implication inverse to the one in \cref{cor:bipartitealgotransfers} cannot hold.

} %
\subsection{Proof of the Hardness Part of~\cref{thm:pathfree}}\label{sub:pthardgeneral}
First, we show a lemma that helps us extend the hardness reductions to all graphs in the second part of~\cref{thm:pathfree}.

\sv{\toappendix{\subsection{Omitted Proof of~\cref{sub:pthardgeneral}}}}

\lv{\begin{lemma}\label{lem:cross}}
\sv{\begin{lemma}[$\spadesuit$]\label{lem:cross}}
  Let $H$ be a graph without isolated vertices and let $u,v\in V(H)$.
  There exists a graph $Z:=H\times H$ with lists $L$ and $z\in V(Z)$ with $L(z)=\{u,v\}$ such that there are at least two homomorphism $h_u, h_v: (Z, L) \tos H$ such that $h_u(z)=u$ and $h_v(z)=v$.
\end{lemma}

\toappendix{%
\sv{%
  \begin{lemma*}[Restated \cref{lem:cross}]
  Let $H$ be a graph without isolated vertices and let $u,v\in V(H)$.
  There exists a graph $Z:=H\times H$ with lists $L$ and $z\in V(Z)$ with $L(z)=\{u,v\}$ such that there are at least two homomorphism $h_u, h_v: (Z, L) \tos H$ such that $h_u(z)=u$ and $h_v(z)=v$.
\end{lemma*}
}
\begin{proof}
  First, we define $z$ as $(u,v)\in V(Z)$ and set $L(z)$ appropriately.
  We do not restrict other lists of $V(Z)$.
  For vertices $(a,b) \in V(Z)$ we define $h_u\left((a,b)\right):=a$ and $h_v\left((a,b)\right):=b$.
  We verify that $h_u$ is indeed a locally subject list homomorphism.
  Take an edge $(a,b)(c,d)\in E(Z)$.
  We infer that $ac\in E(H)$.
  For $(a,b)\in V(Z)$ there exist $d\in N_H(b)$ as no vertex is isolated.
  Vertex $(a,b)$ is happy as for each $c\in N_H(a)$ we have an edge $(a,b)(c,d)\in E(Z)$ and so vertex $(a,b)$ is happy.
  The proof for $h_v$ is analogous.
\end{proof}
}

\noindent Now, we are ready to prove the hardness part of~\cref{thm:pathfree}. In particular, we will show the following theorem.

\begin{theorem}\label{thm:pt-hard}
Let $H\not\in \mathcal H_{\mathrm{poly}}\cup\{P_3,C_4\}$ be a connected graph.
Let $q$ be the number of vertices in the longest induced path in $H\times H$.
There exists $t\le 14+2q$ such that $\llshom{H}$ cannot be solved in time $2^{o(n)}$ in $n$-vertex $P_t$-free graphs, unless the ETH fails.
\end{theorem}

\begin{proof}
  Let $H\not\in \mathcal H_{\mathrm{poly}}\cup\{P_3,C_4\}$.
  As we stated above, the graph $H$ contains at least one of $K_2^\circ, K_2^{\circ\circ}, K_3, P_4, K_{1,3}$ as an induced subgraph $H'$.
  \sv{\cref{thm:universal} proved the cases when $H \in \{K_{1,3}, P_4, K^{\circ\circ}_2\}$.}%
  \lv{Theorems \ref{thm:K13},~\ref{thm:P4}, and~\ref{thm:K2} proved the cases when $H \in \{K_{1,3}, P_4, K^{\circ\circ}_2\}$.}

First, suppose that $H'\in \{K_{1,3},P_4,K_2^{\circ\circ}\}$.
We show how to adjust the hardness construction for $H'$ to $H$.
Recall that in the hardness reductions for $H'$, the lists of each vertex are of size at most two, and moreover, if they are of size exactly two, they always have a neighbor with a list of size precisely one.
Now, we describe how to modify the hardness construction, called \emph{original}, for $\llshom{H'}$ which is an instance $(G',L')$ to an instance $(G,L)$ that will describe a hardness construction for $\llshom{H}$.

Let $w$ be a vertex of $G'$.
First, consider the case that $|L(w)|=1$ where $\{c_w\}=L(w)$.
Let $X_w:=\left(V(H)\setminus V(H')\right)\cap N_H(c_w)$.
In other words, set $X_w$ represents the neighbors of $c_w$ that are only in $H$ but not in $H'$.
We add $|X_w|$ disjoint coppies of graph $H$ into $G$ with lists $L(y)=\{y\}$ for $y\in V(H)$.
For each $c\in X_w$ we connect $c$ in the $c$-th coppy of $H$ with $w$.
We call $c$ the \emph{contact vertex} of the respective additional gadget.

Now, consider the case that $|L(w)|=2$.
Let $\{a,b\}= L(w)$.
As observed there is a vertex $q_w \in N_{G'}(w)$ such that $|L(q)|=1$ and let $\{c_q\}=L(q)$.
Let $X_w^{ab}:=\left(V(H)\setminus V(H')\right)\cap N_H(a)\cap N_H(b)$.
Let $X_w^a:=\left(V(H)\setminus V(H')\right)\cap N_H(a)\setminus X_w^{ab}$.
Let $X_w^b:=\left(V(H)\setminus V(H')\right)\cap N_H(b)\setminus X_w^{ab}$.
Finally, let $X_w:=X_w^b\cup X_w^a$.
In other words, set $X^{ab}_w$ represents common neighbors of $a$ and $b$ that are only in $H$ and not in $H'$.
Similarly, set $X^a_w$ is composed of the neighbors of $a$ which are not the neighbors of $b$, and the symmetrical is true for set $X^b_w$.

We add $|X_w^{ab}|$ disjoint coppies of graph $H$ into $G'$ with lists $L(y)=\{y\}$ for $y\in V(H)$.
For each $c\in X_w^{ab}$, we connect $c$ in the $c$-th coppy of $H$ with $w$.
We call $c$ the \emph{contact vertex}.
We use Lemma~\ref{lem:cross} and we add $|X_w|$ coppies of graph $Z_w$ into $G'$.
We connect $w$ with a special vertex $z\in V(Z_w)$.
For each $c\in X_w$, we define $L(z):=\{c,c_q\}$ in the $c$-th coppy of $Z_w$.
Again, we say that $c$ is the \emph{contact vertex}.
Consult~\cref{fig:entertaining} for an overview of the construction.

\begin{figure}[h]
\centering
\includegraphics[scale=0.8]{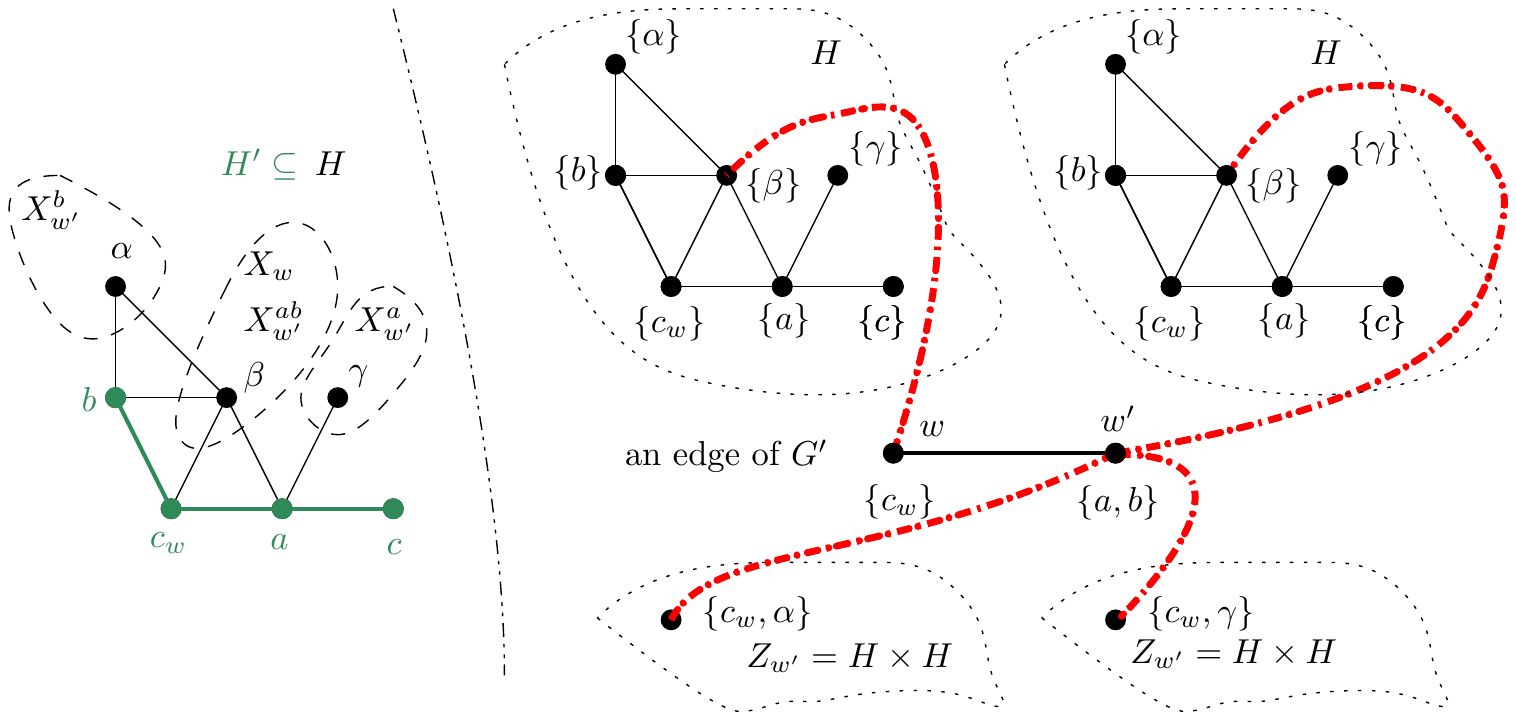} 
\caption{An example of constuction of $(G,L)$ shown on one edge $w,w'$ of $(G',L')$.
  The added gadgets are attached using red dash-dotted edges.
}
\label{fig:entertaining}
\end{figure}

The additional gadgets always allow us to make the original vertices of the construction happy with respect to vertices outside $H'$.
On the other hand, the additional gadgets never allow the contact vertices to be mapped to any vertex of $H$ except for the ones that are already seen in the neighborhood within the original construction.
The exception happens only when the original vertex had a list of size two and the color in the list had a private neighbor (with respect to the other color) outside of $H'$.
Then, the additional gadget may let to map its contact vertex to the color of the neighbor in $G'$ which has the list of size one.
Thus, that does not affect the correctness of the reduction.
Therefore, we conclude that the hardness of the original construction is preserved.

Observe that the length of the longest induced path in $G$ is the length of the longest path in $G'$ plus twice the length of the longest path in $H\times H$, which we denoted as $q$.

It remains to show what to do if $H'\in \{K_2^\circ, K_3\}$ is an induced subgraph of $H$.
As $H$ is non-bipartite, we create a connected bipartite graph $H^*$.
Now, observe that if $H' = K_2^\circ$ then $H^*$ contains $P_4$.
Further, if $H' = K_3$ then $H^*$ contains $C_6$ and so $P_4$.
As in the case of $H^*$ containing a $P_4$, we already proved the hardness, the hardness for $H$ follows by \cref{cor:bipartitetogeneral}.
\end{proof}

\toappendix{
\section{Excluding General Graphs $F$}\label{sec:otherhard}
In this section, we prove \cref{thm:OtherHardness}. Actually, we will work in a slightly more general setting.
For $p,q \geq 1$, let $\cG_{q,p}$ be the class of graphs $G$ such that:
\begin{enumerate}
 \item $\Delta(G) \leq q$,
 \item the girth of $G$ is at least $p$.
\end{enumerate}
We will show the following theorem.

\begin{theorem}
\label{thm:OtherHardness-real}
For every $H \notin \cHpoly$ and every $p \geq 1$ there is $q$ such that 
the \llshom{$H$} problem cannot be solved in time $2^{o(n)}$ on $n$-vertex instances $(G,L)$ where $G \in \cG_{q,p}$, unless the ETH fails.
\end{theorem}

Note that for every $F$ which is not a forest of maximum degree at most $q$, every graph in $\cG_{q,|V(F)|+1}$ is $F$-free.
Thus, \cref{thm:OtherHardness-real} implies \cref{thm:OtherHardness}.

Let $H \not \in \cHpoly$ be a connected graph.
First, we argue that to prove \cref{thm:OtherHardness-real} we may assume that $H$ contains $P_3$ as a subgraph (not necessarily induced).
If $|V(H)| \geq 3$, then $H$ certainly contains $P_3$ as a subgraph as we suppose that $H$ is connected.
The only graphs $H$ such that $|V(H)| \leq 2$ and $H \not \in \cHpoly$ are $K^\circ_2$ and $K^{\circ \circ}_2$.
However, $H^*$ for $H \in \{K^\circ_2,K^{\circ \circ}_2\}$ is connected bipartite graph on 4 vertices and again it contains $P_3$ as a subgraph.
To conclude by \cref{cor:bipartitetogeneral}, it suffices to show hardness for $H^*$.

In the proofs, we will always denote the vertices of a certain $P_3$ of $H$ used in the reductions as $1,2$, and $3$, where $1$ and $3$ are the endpoints of the $P_3$ and $2$ is the middle vertex. 
First, we show the reduction for $H = P_3$.
Further, we use this reduction as a basic step to show hardness for arbitrary $H \notin \cHpoly$.

\subsection{$H = P_3$}
First, we show a slightly stronger version of \cref{thm:OtherHardness-real} for the case $H = P_3$.

\begin{theorem}
 \label{thm:OtherHardnessP3}
 Let $p \in \N$.
 Assuming ETH, the  \llshom{$P_3$} problem cannot be solved in time $2^{o(n)}$ on $n$-vertex instances $(G,L)$, even if they satisfy the following properties:
 \begin{enumerate}
  \item $G \in \cG_{3,p}$.  
  \item The distance between every two vertices of $G$ of degree 3 is at least $p$.
  \item Any list $L(v)$ equals to $\{1,3\}$ or $\{2\}$. Moreover, every edge of $G$ has one endvertex with the list $\{1,3\}$ and the other with the list $\{2\}$.
 \end{enumerate}
\end{theorem}

\begin{proof}
We show a reduction from the \nae problem, where we ask for a truth assignment, where each clause contains at least one true and at least one false literal. 
Let $\Phi$ be an instance of \nae with variables $x_1, \dots,x_n$ and clauses $c_1,\dots,c_m$. 
The ETH implies that there is no algorithm solving every such instance in time $2^{o(n+m)}$~\cite{Lovasz73}.

For every variable $x \in \{x_1,\dots,x_n\}$ we introduce a variable gadget $\vrb(x)$.
It is a cycle of length $r \cdot (d + 1)$, where $d + 1 = 4p$ and $r$ is the number of all occurrences of the variable $x$ in the formula $\Phi$.  
In every variable gadget we mark out $r$ vertices $u^1,\dots,u^{r}$, called \emph{heads}, such that between $u^i$ and $u^{i+1}$, where $i = 1,\dots,r-1$, there are $d$ another vertices $u^i_1, u^i_2,....,u^i_{d}$. 
Analogically, between $u^{r}$ and $u^1$ there are $d$ vertices $u^{r}_1, \dots,u^{r}_{d}$. 
We will refer to vertices between the heads as \textit{sections} of the variable gadgets.
An example of a variable gadget is shown on \cref{fig:cls_gadget}.

For every clause $c \in \{c_1,\dots,c_m\}$ we create a subdivided claw $S_{a_1,a_2,a_3}$ as a clause gadget $\cls(c)$.
Every arm of the subdivided claw corresponds to one literal of $\ell_1,\ell_2,\ell_3$ of the clause $c$ and its length depends on whether the literal is positive or negative.
If the literal $\ell_i$ is positive, then the length $a_i$ of the corresponding arm equals to $4p$, and we will refer to such arms as \textit{positive arms}.
Otherwise, if $\ell_i$ is negative, the length of the arm $a_i$ equals to $4p + 2$ and, analogously, we will refer to them as \textit{negative arms}.

Using these two types of gadgets, we construct the final graph. 
Let $\ell_{x}$ denote a literal equal to a variable $x$ or to its negation and consider a clause $c = (\ell_{x_i}, \ell_{x_j}, \ell_{x_k})$.
If a variable $x$ of $c$ occurs as a positive literal, then we connect one of the head vertices $\{u^1,\dots,u^r\}$ of the gadget $\vrb(x)$ with the endpoint of one positive arm of $\cls(c)$. 
Otherwise, we use a negative arm. 
We connect every head of each variable gadget with exactly one endpoint of an arm of each clause gadget. 
An example of a clause gadget connected to the three gadgets of the corresponding variables is shown on \cref{fig:cls_gadget}.
This completes the construction of the graph $G$.

Observe that $G$ is bipartite. Indeed, each variable gadget is an even cycle, and all heads are in one bipartition class.
Furthermore, each arm of a clause gadget has an even length, which means that the center of each clause gadget is in the different bipartition class than the heads of variable gadgets.

\begin{figure}[h]
\centering
\includegraphics{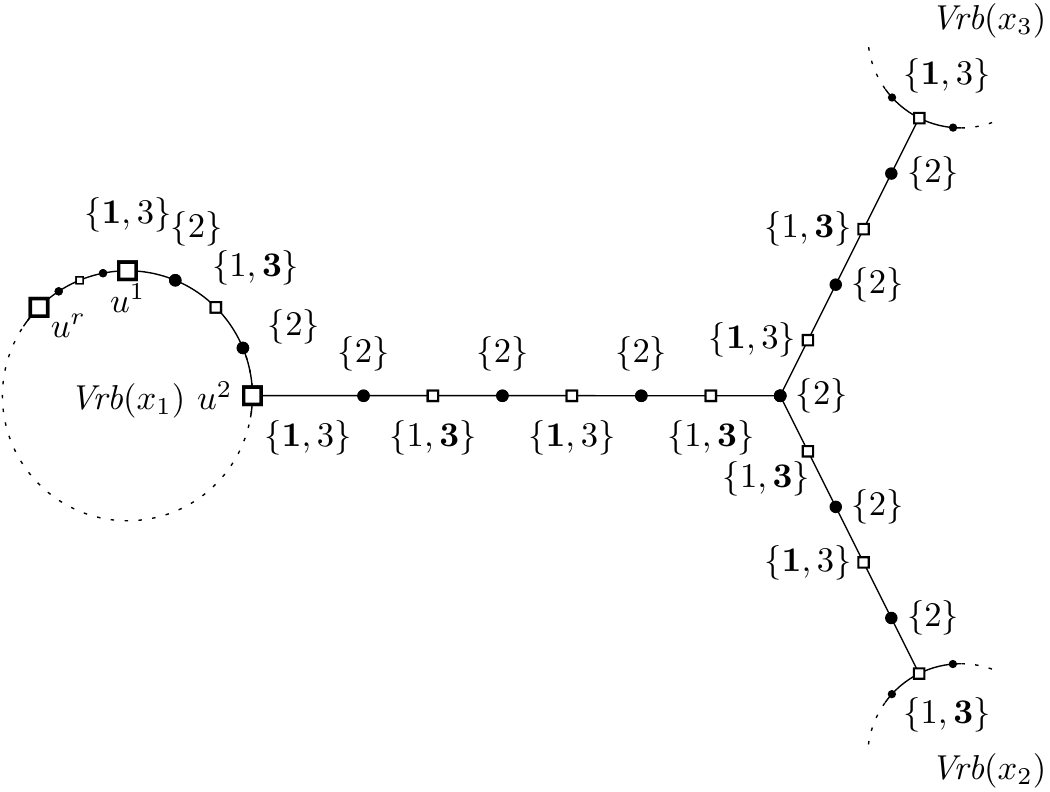} 
\caption{An example of variable gadgets connected to a clause gadget $\cls(c)$ for $p = 1$ and $c = \neg x_1 \vee x_2 \vee x_3$ (thus, $\cls(c)$ has two positive and one negative arm). 
Since the graph $G'$ is bipartite, vertices of one bipartition class (with the list $\{1,3\}$) are depicted as boxes and vertices from the other one as disks.
The bold numbers in the lists $\{1,3\}$ represent images of the vertices under a homomorphism $h: (G,L) \tos P_3$ created from a truth assignment of $\Phi$ such that $x_1$ and $x_3$ are true and $x_2$ is false.} 
\label{fig:cls_gadget}
\end{figure}

We argue the graph $G$ has the sought properties.
It is clear that $\Delta(G) = 3$ and only vertices of degree 3 are the centers of the clause gadgets and the heads of the variable gadgets.
The distance between two consecutive heads of a variable gadget is $4p$, and the distance between a head and a center of a clause gadget connected by an arm is at least $4p$ as well.
Thus, each vertex of degree 3 is separated by at least $4p$ vertices of degree 2.
The last property to check is that $G$ does not have a short cycle.
Clearly, cycles $\vrb(x)$ are longer than $4p$.
Every other cycle in $G$ has to include an arm of some clause gadget.
Observe that since $p$ is a constant, then the number of vertices of the reduction graph is linear in the number of clauses and variables of the formula $\Phi$.

Now, we set the lists $L$. We recall that $1,2,3$ are the vertices of $P_3$, where $2$ is the middle one.
As observed above, $G$ is bipartite, and all heads of variable gadgets are in one bipartition class $X$, while the centers of clause gadgets are in the other bipartition class $Y$. We assign the list $\{1,3\}$ to all vertices from $X$ and the list $\{2\}$ to all vertices of $Y$.

Finally, we show there is a homomorphism $h:(G,L) \tos P_3$ if and only if $\Phi$ is \yes-instance of \nae.
Let $\phi$ is a satisfying truth assignment of $\Phi$, and we construct the sought homomorphism $h: (G,L) \tos P_3$.
All vertices with the list $\{2\}$ are mapped to the vertex $2$.
Since all vertices with the list $\{1,3\}$ are connected to a vertex with the list $\{2\}$, the vertices with the list $\{1,3\}$ are trivially happy.
We say that a homomorphism $h$ assigns values to some sequence of a vertices $(w_1,\dots,w_s)$ according to a \textit{pattern} $(p_1,\dots,p_t)$, where $t \leq s$, if $h(w_i) = p_{(i  \mod t)+1}$ for $i = 1,\dots,s$, see \cref{fig:pattern} for an example.

\begin{figure}[h]
\centering
\includegraphics{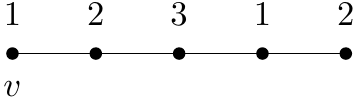}
\caption{A path with values of vertices assigned according to the pattern $(1,2,3)$ starting from the vertex $v$.} 
\label{fig:pattern}
\end{figure}

Let $x$ be a variable of $\Phi$.
If $\phi(x)$ is true, then we set $h(u)=1$ for every head $u \in \{u^1,u^2,...,u^{r}\}$ in the variable gadget $\vrb(x)$. 
For each section $u^i_1,\dots,u^i_{d}$ in the gadget $\vrb(x)$, the homomorphism $h$ assigns the values according to the pattern $(2,3,2,1)$.
Since there are $4p-1$ section vertices between two consecutive heads in $\vrb(x)$, all vertices with the list $\{2\}$ are happy. 
Note that the last but one vertex $u^i_{d-1}$ of the section $u^i_1,\dots,u^i_{d}$ is assigned the value $3$ and the head $u^{i+1}$ has the value $1$, thus the last vertex $u^i_d$ of this section is indeed happy.
Let $u$ be a head vertex of the gadget $\vrb(x)$ and $v_1,\dots,v_z$ be vertices of the arm connected to $u$, where $v_1$ is connected to $u$ and $v_z$ is connected to the center of the corresponding clause gadget.
We assign the value of $h$ to $v_1,\dots,v_z$ again according to the pattern $(2,3,2,1)$.

On the other hand, if $\phi(x)$ is false, the assigning values are analogous, but we switch $1$ and $3$.
I.e, we set $h(u) = 3$ for all heads $u$ in the gadget $\vrb(x)$ and we assign the values to the section vertices and the vertices of arms of clause gadgets according to the pattern $(2,1,2,3)$.
See \cref{fig:cls_gadget} for an illustration of the construction of $h$.

We set the value of $h$ to all vertices of $G$.
It is clear that $h$ is a homomorphism, and it remains to argue that the centers of all clause gadgets are happy, as we argue that for all other vertices.
Let $w$ be a center of a clause gadget $\cls(c)$.
Since $\Phi$ is satisfiable instance of \nae with the satisfying assignment $\phi$, the clause $c$ contains a literal $\ell_1$ with a value true (under $\phi$) and a literal $\ell_2$ with a value false.

Let $\ell_1$ be a positive occurrence of a variable $x$.
Thus, $\phi(x)$ is true.
Then, for all heads $u$ of the gadget $\vrb(x)$ holds that $h(u) = 1$. 
Let $v_1,\dots,v_z$ be vertices of the arm connected a head $u$ of $\vrb(x)$ and the center $w$ of $\cls(c)$ ($v_q$ is connected to $w$).
By construction of $h$, the values of $v_1,\dots,v_z$ are assigned according to the pattern $(2,3,2,1)$.
Since $\ell_1$ is a positive literal, then $z = 4p$ and $h(v_z) = 1$.

On the other hand, if $\ell_1$ is a negative occurrence of $x$, then $\phi(x)$ is false and for all heads $u$ of $\vrb(x)$ holds that $h(u) = 3$.
The arm $v_1,\dots,v_z$ is negative and $z = 4p + 2$.
The pattern used for $v_1,\dots,v_z$ is now $(2,1,2,3)$.
Thus, $h(v_z) = 1$ again.

The case for $\ell_2$ is analogous; however we switch $1$ and $3$.
Thus, the vertex $v'_z$ of the corresponding arm connected to $w$ has value $h(v'_z) = 3$, independently on whether $\ell_2$ is a positive or negative literal.
Therefore, we can conclude that the vertex $w$ is happy as it is connected to vertices $v_z$ and $v'_z$ such that $h(v_z) = 1$ and $h(v'_z) = 3$.

Now, suppose there is a homomorphism $h: (G,L) \tos P_3$.
Let $x$ be a variable of $\Phi$ and $u^1,\dots,u^r$ be heads of the variable gadget $\vrb(x)$.
Since there are $4p-1$ section vertices between $u^i$ and $u^{i+1}$ (and $u^r$ and $u^1$) and $h$ is a locally surjective, then $h$ has to assign the same value to all heads $u^1,\dots,u^r$.
If $h$ assigns 1 to the heads of $\vrb(x)$, then we set $\phi(x) = 1$. 
Otherwise, $\phi(x) = 0$. 
Let us show that $\phi$ satisfies $\Phi$. 
Suppose that there is a clause $c$ which is not satisfied by $\phi$. 
It means that all literals occurring in $c$ are all true or all are false.
Suppose that all literals in $c$ are true (as the other case is analogous) and let $\ell$ be a literal of $c$.

Let $\ell$ be a positive occurrence of a variable $x$ in $c$.
Let $u$ be a head of $\vrb(x)$ connected by an arm $v_1,\dots,v_z$ to the center $w$ of the clause gadget $\cls(c)$.
Since $\ell$ is a positive literal and it is true under $\phi$, it holds that $h(u) = 1$ and $z = 4p$.
Thus, it must hold that $h(v_z) = 1$ as $h$ is locally surjective.

On the other hand, if $\ell$ is a negative occurrence of $x$, then $h(u) = 3$ and $z = 4p + 2$.
However, it holds $h(v_z) = 1$ again in this case as well.
Since we suppose that all literals of $c$ are true (under $\phi$), the vertex $w$ is connected only to vertices $v$ such that $h(v) = 1$.
Thus, the homomorphism $h$ would not be locally surjective.

The case when all literals of $c$ are false is analogous.
We only switch the role of $1$ and $3$, i.e., the center vertex $w$ of $\cls(c)$ would be connected only to vertices $v$ such that $h(v) = 3$.
Thus, we conclude that there is a homomorphism $h: (G,L) \tos P_3$ if and only if the formula $\Phi$ is a satisfiable instance of \nae, which completes the proof of \cref{thm:OtherHardnessP3}.
\end{proof}

We observed that if $F \notin \cS$, then every graph in $\cG_{3,|V(F)|+1}$ is $F$-free.
Thus \cref{thm:OtherHardnessP3} immediately yields the following corollary, complementing \cref{thm:algo-p3-forests}.

\begin{corollary}
For every $F \notin \cS$, the \llshom{$P_3$} problem cannot be solved in subexponential time in $F$-free graphs unless the ETH fails.
\end{corollary}

\subsection{General $H$}
In this section, we will finish the proof of \cref{thm:OtherHardness-real}. 
Recall that we may assume that $H$ contains $P_3$ as a subgraph, as otherwise, by \cref{cor:bipartitetogeneral},  we can consider $H^*$ instead.
We select a fixed $P_3$ in $H$, and we start with the reduction from the previous section.
However, the homomorphisms used in the previous proof are not locally surjective anymore, as the vertices $1,2$, and $3$ of the selected $P_3$ may have other neighbors in $H$ apart from their neighbors in the $P_3$.
We will fix this by new gadgets, which we will connect to the vertices of the original graphs to make them happy again.

Recall that there are only two types of lists ($\{1,3\}$ and $\{2\}$) used in the reduction.
We say the vertex $2$ of $H$ is \emph{trivial} if the vertex $2$ does not have any other neighbor in $H$ apart from $1$ and $3$.
Analogously, the vertex $w \in \{1,3\}$ is \emph{trivial} if $w$ are connected only to the vertex $2$ in $H$.
We would like to point out that if $H = K_3$, then the vertex $2$ is trivial, and $1$ and $3$ are not trivial as they are connected to each other.

Note that if $2$ is trivial, then any locally surjective homomorphism $h: (G,L) \tos P_3$ makes all vertices with the list $\{2\}$ happy even if we consider $h$ as homomorphism (not necessarily locally surjective) $(G,L) \to H$.
Analogously, it holds for vertices with the list $\{1,3\}$ if both vertices $1$ and $3$ are trivial.

We will construct two types of gadgets, each for one type of the lists.
Each such gadget will have a special vertex $v$, which will be identified with an original vertex $u$ (thus, $v$ and $u$ will have the same lists).
We will construct a graph $G'$ by appending such gadgets to each vertex of the original graph $G$ from the previous reduction.
To finish the proof, we will prove there is a locally surjective homomorphism $h: (G,L) \tos P_3$ if and only if there is a locally surjective homomorphism $h': (G',L') \tos H$.
Now, we present the gadgets more formally.

\begin{lemma}
 \label{lem:Gadget2}
 Let $H$ be a connected graph containing $P_3$ (with the vertices $1,2,3$) as a subgraph.
 Let $S$ be a list $\{1,3\}$ or $\{2\}$.
 Suppose at least one vertex in $S$ is not trivial.
 Then, for any $p \geq |V(H)|^2$, there is a graph $H_S$ with lists $L_S$ of the following properties:
 \begin{enumerate}
  \item $H_S \in \cG_{q,p}$ where $q = f(\Delta(H), p)$ for a suitable function $f$. 
  \item There is a vertex $v \in V(H_S)$, called \emph{root}, such that $L_S(v) = S$.
  \item For each $i \in S$, there is a list homomorphism $h_i: (H_S, L_S) \to H$ such that $h_i(v) = i$.
  \item If $S = \{1,3\}$, then the homomorphisms $h_1$ and $h_3$ make  all vertices of $H_S$ happy (i.e., the homomorphism $H_S$ is locally surjective).
  \item If $S = \{2\}$, then:
  \begin{enumerate}
   \item Any neighbor $u$ of $v$ does not contain $1$ and $3$ in its lists. 
   \item The homomorphism $h_2$ makes  all vertices of $H_S$, except the root $v$, happy. 
  Moreover, $h_2(N_{H_S}(v)) = N_H(2) \setminus \{1,3\}$.
  \end{enumerate}
 \end{enumerate}
\end{lemma}
The last sought property says the following.
Suppose we identify the root $v$ of gadget $H_2$ to a vertex $u \in V(G)$ with a list $\{2\}$ from the original reduction graph $G$ and denote the new graph $G'$.
Let $h: (G,L) \to P_3$ be a list homomorphism that makes the vertex $u$ happy.
Now, if we combine $h$ with $h_2$ to a homomorphism $h': (G',L') \to H$, then $h'$ makes the vertex $v = u$ happy as well.
On the other hand, if $h$ does not make the vertex $u$ happy, then there is no list homomorphism $h'_2: (H_2,L_2) \to H$ that combined with $h$ would make the vertex $u$ happy as any neighbor of $v = u$ in $H_2$ does not contain $1$ and $3$ in its list.

We will prove \cref{lem:Gadget2} by analyzing two cases according to two types of lists.
We start with the case $S = \{2\}$, as its proof is slightly easier.
Further, we use a similar idea to prove the lemma for $S = \{1,3\}$.

\begin{proof}[Proof of \cref{lem:Gadget2} for $S = \{2\}$]
 First, we start to build the graph $H_2 := H_{\{2\}}$ as a tree using BFS-like procedure on the vertices of $H$.
 We start with its root $v$ and we set the list $L_2(v) = \{2\}$.
 Let $N'(2) = N_{H}(2) \setminus \{1,3\}$.
 Note that $N'(2)$ is not empty as we suppose the vertex $2$ is not trivial.
 We add to $H_2$ copies of vertices of $N'(2)$ as children of $v$.
 Thus, each added vertex $u'$ to $H_2$ is a copy of a vertex $u$ of $H$ and we can set each list $L_2(u') = \{u\}$.
 In this way, we set lists of all vertices which we will add.
 
 Now, we will continue recursively in the building of $H_2$.
 Let $u'$ be a leaf of $H_2$, and it is a copy of a vertex $u \in V(H)$.
 We add copies of vertices $N_H(u)$ to $H_2$ as children of $u'$.
 We do not mark vertices whose copies are already in $H_2$, and a vertex of $H$ will have several copies of itself in $H_2$.
 Moreover, if we add copies of neighbors of the vertex $2$ again, we also include the vertices $1$ and $3$, not like in the first level of the tree where we exclude them.
 
 We end our building of $H_2$ when all the leaves of $H_2$ are at depth $2p$.
 Note that all lists have size 1, and the only possible homomorphism $h_2: (H_2,L_2) \to H$ makes all vertices of $H_2$ happy, except the leaves and the root.
 Moreover, for the root it holds that $h_2(N_{H_2}(v)) = N_H(2) \setminus \{1,3\}$.
 Thus, we add edges to the leaves and certain vertices in $H_2$ to make them happy as well.
 
 Let $u'$ be a leaf of $H_2$ and copy of $u \in V(H)$.
 Since we assume that $p \geq |V(H)|$, each vertex of $H$ has a copy in the first $p$ levels of $H_2$.
 Thus, we pick such a copy $w'$ (for $w'$ not being the root $v$) for each vertex $w \in N_H(u)$ and add an edge $u'w'$ to $E(H_2)$.
 We repeat this procedure for each leaf $u$ of $H_2$.
 This finishes the construction of $H_2$.
 Each vertex except the root is happy now in $h_2$.
 Since we did not add any new edge incident to the root $v$, it still holds that $h_2(N_{H_2}(v)) = N_H(2) \setminus \{1,3\}$.
 See \cref{fig:funny_gadget}, for an example of the construction of $H_2$.
 
 \begin{figure}[h]
  \centering
  \includegraphics{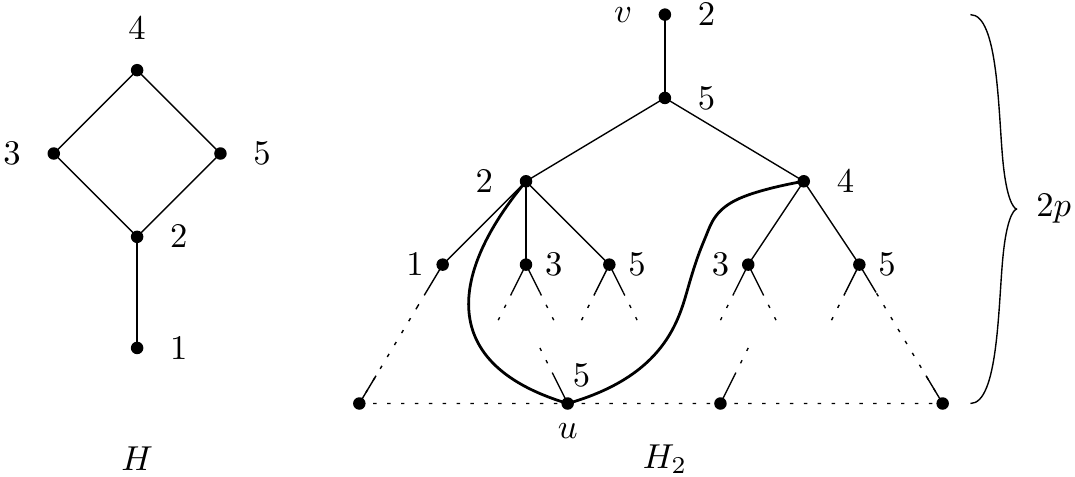}
  \caption{Example of the construction of the gadget $H_2$ with few top levels, some leaves, and the root $v$. The thicker edges from the leaf $u$ depicted the edges added after the BFS-like procedure to make the leaf $u$ happy in the homomorphism $h_2$.}
  \label{fig:funny_gadget}
 \end{figure}

 We connect the leaves only to the vertices in the first $p$ levels of $H_2$.
 Since the distance between any vertex in the first $p$ levels and any leaf (in the level $2p$) is at least $p$, we did not make a cycle shorter than $p$.
 Moreover, the maximum degree of $H_2$ depends only on $p$ and $\Delta(H)$ (it is roughly bounded by $\Delta(H)^{2p}$).
 Thus, $H_2 \in \cG_{q,p}$ for a suitable $q$. 
\end{proof}

We are aware that for some $H$, there is a simpler construction of $H_2$ with fewer vertices and a smaller maximum degree.
However, we described a general construction that is as simple as possible for the sake of the readability of the proof.
Now, we will proceed with the proof for the case $S = \{1,3\}$.
We will construct $H_{1,3}$ again by BFS-like procedure.
However, we will run the procedure on $H \times H$ instead of $H$.

\begin{proof}[Proof of \cref{lem:Gadget2} for $S = \{1,3\}$]
 By assumption of \cref{lem:Gadget2}, we know that at least one of vertices $1$ and $3$ has a neighbor apart from $2$.
 Let $K = H \times H$, i.e.,
 \begin{align*}
  V(K) &= V(H) \times V(H), \\
  E(K) &= \bigl\{\{(u_1,v_1),(u_2,v_2)\} \mid u_1u_2,v_1v_2 \in E(H) \bigr\}.
 \end{align*}
 We run again BFS-like procedure to build $H_{1,3}:=H_{\{1,3\}}$, now on the graph $K$.
 We start from the vertex $(1,3) \in V(K)$ and make its a copy as the root $v$ of $H_{1,3}$.
 The vertices of $V(K)$ correspond to list of vertices of $H$ of length at most 2, i.e., a vertex $u'$ that is a copy of $(u_1,u_2) \in V(K)$ will have the list $(u_1,u_2)$.
 For the purpose of this construction, we consider these lists as ordered pairs of vertices of $H$.
 Thus, the root $v$ has the list $(1,3)$.
 We run the BFS-like procedure. Moreover, we do not need any exception for the root as we needed in the previous case.
 To a leaf $u'$ of  $H_{1,3}$ that is a copy of $u = (u_1,u_2) \in V(K)$ we add copies of vertices $N_K(u) = N_H(u_1) \times N_H(u_2)$.
 We stop the process again when all leaves of $H_{1,3}$ are in depth $2p$.
 We define the homomorphisms $h_1, h_3: (H_{1,3},L_{1,3}) \to H$ naturally as follows.
 Let $u' \in V(H_{1,3})$ with a list $(u_1,u_3)$.
 Then, $h_1(u') = u_1$ and $h_3(u') = u_3$.
 
 First, we argue that $h_1$ and $h_3$ are homomorphisms, and they make all vertices happy except the leaves.
 Let $u'w' \in E(H_{1,3})$ and $u'$ be a copy of $u = (u_1,u_3) \in K$ and $w'$ be a copy of $w = (w_1,w_3)$.
 Then, $h_i(u') = u_i$ and $h_i(w') = w_i$.
 By construction of $H_{1,3}$, it holds that $uw \in E(K)$.
 It follows that $u_1w_1 \in E(H)$ and $u_3w_3 \in E(H)$.
 Thus, $h_1$ and $h_3$ are indeed homomorphisms.
 Suppose $u'$ is a non-leaf vertex of $H_{1,3}$.
 Thus, the set of lists of all children of $u'$ is the set $N_H(u_1) \times N_H(u_3)$.
 Therefore, $h_i(N_K(u')) = N_H(u_i)$, and the homomorphism $h_i$ makes the vertex $u'$ happy.
 
 To finish the construction of $H_{1,3}$ we add some edges to the leaves, in a similar way as we did in the proof for the previous case, that the homomorphisms $h_1$ and $h_3$ make happy all vertices of $H_{1,3}$.
 Let $C$ be a connected component of $K$ that contains the vertex $(1,3)$.
 Note that the graph $H_{1,3}$ contains only copies of vertices of $C$.
 Since $p \geq |V(H)|^2$, each vertex of $C$ has at least one copy in the first $p$ levels of $H_{1,3}$.
 Let $u'$ be a leaf of $H_{1,3}$ and a copy of $u \in K$.
 For each vertex $w \in N_K(u)$ we take one of its copy $w'$ in the first $p$ levels of $H_{1,3}$ and add an edge $\{u',w'\}$ to $H_{1,3}$.
 This finishes the construction of $H_{1,3}$.
 Thus, the homomorphisms $h_1$ and $h_3$ make happy all vertices of $H_{1,3}$.
 Again, the graph $H_{1,3}$ does not contain a cycle shorter than $p$ and the maximum degree of $H_{1,3}$ can be bounded by a function of $\Delta(H)$ and $p$ (roughly by $\Delta(H)^{4p}$). 
\end{proof}

\noindent Now, we are ready to finish the proof of \cref{thm:OtherHardness-real}.

\begin{proof}[Proof of \cref{thm:OtherHardness}]
 Let $P_3$ be a subgraph of $H$ with the vertices $1,2$ and $3$.
 We start with the reduction from \nae to \llshom{$P_3$} presented in the proof of \cref{thm:OtherHardnessP3}.
 Let $\Phi$ be an instance of \nae and $(G,L)$ be the constructed graph with lists.
 We will append the gadgets given by \cref{lem:Gadget2} to $G$ to get a graph with lists $(G', L')$ in such a way there is a locally surjective homomorphism $h: (G,L) \tos P_3$ if and only if there is a locally surjective homomorphism $h':  (G',L') \tos H$ -- we call the appended gadgets as \emph{funny gadgets} because they make the original vertices of $G$ happy again.
 
 Let $u \in V(G)$.
 Recall that the lists in $L$ are only of type $\{1,3\}$ and $\{2\}$.
 First, suppose that $L(u) = \{2\}$.
 If the vertex $2$ is not trivial, then we add a copy $H^u_2$ of $H_2$ given by \cref{lem:Gadget2} and identify $u$ with the root of the added copy.
 Similarly, if $L(u) = \{1,3\}$ and at least one of the vertices $1$ and $3$ is not trivial, then we add a copy $H^u_{1,3}$ of $H_{1,3}$ and identify its root and $u$. 
 We repeat this for all vertices of $G$, that finishes the construction of $G'$.
 See \cref{fig:funny_gadget_usage}, for an example of appending the funny gadgets to the graph $G$.
 \begin{figure}[h]
  \centering
  \includegraphics{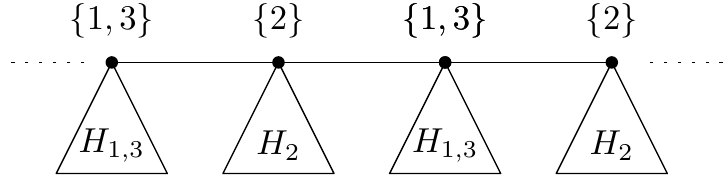}
  \caption{Usage of the funny gadgets $H_2$ and $H_{1,3}$ on a subpath of the graph $G$. To each vertex $v$ of $G$ with a list $S$ we append a copy of the funny gadget $H_S$.}
  \label{fig:funny_gadget_usage}
 \end{figure}

 Note that when we append the funny gadgets, we identify only vertices with the same lists.
 Thus, the new lists for the original vertices of $G$ remain the same, and the lists of vertices the funny gadgets are given by \cref{lem:Gadget2}. 
 By \cref{lem:Gadget2} and \cref{thm:OtherHardnessP3}, the maximum degree of $G'$ is bounded by a function of $\Delta(H)$ and $p$.
 By construction, any funny gadget is connected to $G$ through only one vertex.
 Thus, any cycle of $G'$ cannot simultaneously contain both edges of $G$ and edges of any funny gadget.
 Therefore, we conclude that $G'$ does not contain any cycle shorter than $p$.
 Note that the size of $G'$ is linear in the number of variables and clauses of $\Phi$ as the size of any funny gadget is constant.
 
 It remains to prove that there is a list locally surjective homomorphism ${h: (G,L) \tos P_3}$ if and only if there is a list locally surjective homomorphism ${h': (G',L') \tos H}$.
 It follows that $\Phi$ is satisfiable if and only if there is a list locally surjective homomorphism $h': (G',L') \tos H$, by proof of \cref{thm:OtherHardnessP3}.
 Suppose there is a homomorphism $h: (G,L) \tos P_3$.
 We extend $h$ to $h'$ to be a locally surjective homomorphism $(G',L') \tos H$.
 Let $u \in V(G)$ and $L(u) = \{2\}$.
 Suppose the vertex $2$ is not trivial; otherwise, there is no funny gadget appended to the vertex $u$.
 We define $h'(w) = h_2(w)$ for any vertex $w$ of the funny gadget $H^u_2$ (note that $h(u) = h_2(v) = 2$ for the root $v = u$ of $H^u_2$).
 Similarly if $L(u) = \{1,3\}$, we set $h'(w) = h_i(w)$ for any vertex $w$ of the gadget $H^u_2$, where $i = h(u) \in \{1,3\}$ (again only in the case when at least one of the vertices $1$ and $3$ is not trivial).
 
 It is clear that $h': (G,L) \to H$ is a list homomorphism.
 By \cref{lem:Gadget2}, the homomorphism $h'$ makes happy all non-root vertices of all funny gadgets.
 All vertices of the original graph $G$ with the list $\{1,3\}$ are happy by \cref{lem:Gadget2}, if the vertices $1$ or $3$ are not trivial, or by the construction of $h$ and $h'$, if both $1$ and $3$ are trivial.
 Now, consider a vertex $u \in V(G)$ with the list $\{2\}$.
 The vertex $u$ has two neighbors in $V(G)$ that are mapped to $1$ and $3$ by the construction of $h$ and $h'$.
 If the vertex $2$ is trivial, then $u$ is clearly happy in $h'$.
 Otherwise by \cref{lem:Gadget2}, we have that $h'(N_{G'}(u) \cap V(H^u_2)) = N_H(2) \setminus \{1,3\}$.
 Therefore, $h'(N_{G'}(u)) = N_H(2)$, i.e., $h'$ makes happy the vertex $u$.
 
 Now, suppose there is a list locally surjective homomorphism ${h': (G',L') \tos H}$.
 Recall that for any vertex $u$ of $G$ with the list $\{2\}$ holds that the vertices $1$ and $3$ are not in the lists of neighbors of $u$ in the funny gadget $H^u_2$.
 Thus, to make the 
 vertex $u$ happy the homomorphism $h'$ has to map some of neighbors of $u$ in $G$ to $1$ and some neighbors of $u$ in $G$ to $3$.
 Any vertex of $G$ with the list $\{1,3\}$ has a neighbor in $G$ with the list $\{2\}$.
 Therefore, if we restrict $h'$ to the original graph $G$ we get a locally surjective homomorphism $\tilde{h}: (G,L) \tos P_3$.
  \end{proof}

}

\section{Concluding Remarks}
Let us conclude the paper with discussing some potential ways to strengthen our results.
First, recall that in the hardness part of \cref{thm:pathfree}, the length $t$ of the forbidden induced path depends on $H$.
One might wonder if it is possible to find $t$, such that for every $H \notin \cHpoly \cup \{P_3,C_4\}$, 
the \llshom{$H$} problem is hard in $P_t$-free graphs.

Suppose that such a $t$ exists and consider $H = P_{t}$ with consecutive vertices $1,\ldots,t$.
Without loss of generality, we may assume that $t \geq 4$.
Consider a locally surjective homomorphism $h$ from $G$ to $P_{t}$.
Note that $h$ is in particular surjective, so there exists a vertex $v_1$ mapped to $1$.
By local surjectivity of $h$, there must be a neighbor $v_2$ of $v_1$ mapped to 2, a neighbor $v_3$ of $v_2$ mapped to 3, and so on.
Note that $v_1,v_2,\ldots,v_{t}$ is a path in $G$. Furthermore, this path is induced, as otherwise $h$ is not a homomorphism.
Consequently, every \yes-instance of \lshom{$P_{t}$} (and thus of \llshom{$P_{t}$}) contains an induced $t$-vertex path.
This means that \llshom{$P_{t}$} is polynomial-time solvable (and actually trivial) in $P_t$-free graphs.
On the other hand, $P_t \notin \cHpoly \cup \{P_3,C_4\}$, so by \cref{thm:pathfree}~(2.) there exists some $t'$, for which the problem is hard in $P_{t'}$-free graphs. 

Second, recall that in \cref{thm:OtherHardness} the degree bound on $F$ depends on $H$.
Again, one might wonder if this is necessary. However, every \yes-instance of \lshom{$H$} must contain a vertex of degree $\Delta(H)$,
as some vertex $v$ of $G$ must be mapped to a maximum-degree vertex $a$ of $H$, and all vertices from $N_H(a)$
must appear on the set $N_G(v)$. Consequently, we cannot hope for a universal upper bound on the degree of $G$ in the proof of \cref{thm:OtherHardness}.

The above two examples show that obtaining the full characterization of pairs $(H,F)$, for which \llshom{$H$} admits 
a subexponential-time algorithm in $F$-free graphs, would be a tedious task.
One can probably start with some small graphs $F$. Let us point out that if $F=P_4$, then \llshom{$H$} is polynomial-time solvable for every $H$.
Indeed, $P_4$-free graphs, also known as cographs, have bounded \emph{cliquewidth} and the result follows from the celebrated meta-theorem for bounded-cliquewidth graphs by Courcelle, Makowsky, and Rotics~\cite{CourcelleMR00}.

\sv{\bibliographystyle{abbrv}}
\lv{\bibliographystyle{plainurl}}
\bibliography{main}

\sv{
  \newpage
\appendix
\appendixText
}

\end{document}